\documentclass[journal, twocolumn]{IEEEtran}

\usepackage{comment,cite}
\usepackage{amsmath, amssymb, bm}
\usepackage{balance}

\usepackage{microtype}
\usepackage{graphicx}
\usepackage{subfigure}
\usepackage{booktabs} 

\let\emptyset\varnothing

\usepackage{mathtools}

\DeclarePairedDelimiter\floor{\lfloor}{\rfloor}

\newcommand{\myapprox}{{\raise.17ex\hbox{$\scriptstyle\sim$}}}
\newcommand\norm[1]{\left\lVert#1\right\rVert}

\newcommand{\icol}[1]{
  \left[\begin{matrix}#1\end{matrix}\right]%
}

\newcommand\mydots{\hbox to 1em{.\hss.\hss.}}
\newcommand{\dx}[1]{\textrm{d}#1}

\usepackage{xcolor}
\usepackage{braket, amsthm, amsfonts}
\usepackage[ruled,longend]{algorithm2e}
\usepackage{hyperref}

\usepackage{multirow}
\usepackage{enumitem} 

\usepackage[font=normalsize,skip=0pt]{caption}

\makeatletter
\g@addto@macro\normalsize{%
  \setlength\abovedisplayskip{5pt}
  \setlength\belowdisplayskip{5pt}
  \setlength\abovedisplayshortskip{5pt}
  \setlength\belowdisplayshortskip{5pt}
}
\makeatother

\makeatletter
\def\thm@space@setup{%
  \thm@preskip=0.15cm plus 0.05cm minus 0.05cm
  \thm@postskip=0.05cm plus 0.05cm minus 0.05cm
}
\makeatother

\newcommand{\Prob}[1]{\Pr\Bigl\{#1\Bigr\}}

\theoremstyle{plain}
\newtheorem{theorem}{Theorem}
\newtheorem{lemma}{Lemma}

\newtheorem{corollary}{Corollary}

\theoremstyle{definition}
\newtheorem{definition}{Definition}

\theoremstyle{remark}
\newtheorem{remark}{Remark}

\allowdisplaybreaks
\makeatletter
\newcounter{longaligned}

\makeatother

\begin{document}
\title{Evaluating Load Balancing Performance in Distributed Storage with Redundancy}
\author{%
        Mehmet~F.~Akta\c{s},
        Amir~Behrouzi-Far,
        Emina~Soljanin,~\IEEEmembership{Fellow,~IEEE,} and
        Philip~Whiting
        
\thanks{M. F. Akta\c{s}, A. Behrouzi-Far and E. Soljanin are with the Department of Electrical and Computer Engineering, Rutgers University, New Brunswick,
NJ, 08901 USA. (email: \{mehmet.aktas, amir.behrouzifar, emina.soljanin\}@rutgers.edu)}
\thanks{P. Whiting is with the Department of Engineering, Macquarie University, NSW, 2109, Australia. (email: philip.whiting@mq.edu.au)}

\thanks{This paper was presented in part at 2019 XVI International Symposium ``Problems of Redundancy in Information and Control Systems'' \cite{Redundancy:AktasBS19}.}
\thanks{
Copyright (c) 2017 IEEE. Personal use of this material is permitted.  However, permission to use this material for any other purposes must be obtained from the IEEE by sending a request to pubs-permissions@ieee.org.}

}

\maketitle

\begin{abstract}
To facilitate load balancing, distributed systems store data redundantly. 
We evaluate the load balancing performance of storage schemes in which each object is stored at $d$ different nodes, and each node stores the same number of objects. In our model, the load offered for the objects is sampled uniformly at random from all the load vectors with a fixed cumulative value.
We find that the load balance in a system of $n$ nodes improves multiplicatively with $d$ as long as $d = o\left(\log(n)\right)$, and improves exponentially once $d = \Theta\left(\log(n)\right)$.
We show that the load balance improves in the same way with $d$ when the service choices are created with XOR's of $r$ objects rather than object replicas.
In such redundancy schemes, storage overhead is reduced multiplicatively by $r$. However, recovery of an object requires downloading content from $r$ nodes. At the same time, the load balance increases additively by $r$.
We express the system's load balance in terms of the maximal spacing or maximum of $d$ consecutive spacings between the ordered statistics of uniform random variables. Using this connection and the limit results on the maximal $d$-spacings, we derive our main results.
\end{abstract}


\section{Introduction}
\label{sec:intro}
Distributed computing systems are built on a storage layer that provides data write/read service for executing workloads. Thus, the overall performance of a computing system depends on the data access (I/O) performance implemented by the underlying storage system. In production systems, data access times are the main bottleneck to performance \cite{ChallengesInBuildingLargeScaleInformationRetrievalSystems:Dean09}.
Indeed, access times in modern large-scale systems (e.g., Cloud systems) greatly suffer from storage nodes that exhibit poor or variable performance \cite{Dremel:MelnikGL10}. Poor performance is caused by many factors, but primarily it comes from multiple-workload resource sharing and the resulting contention at the system resources \cite{TailAtScale:DeanB13}.
Poor or variable performance is possible at any level of load but it is certainly aggravated at overloaded storage nodes \cite{StragglerRootCauseAnalysisInDatacenters:OuyangGY16}.
It is, therefore, paramount for distributed systems to be able to balance the offered data access load across the storage nodes.

It follows that to achieve good data access performance, we must balance the offered load across the storage nodes as evenly as possible.
In modern storage systems (e.g., HDFS \cite{HDFS:ShvachkoKR10}, Cassandra \cite{Cassandra:LakshmanM10}, Redis \cite{Redis}), data objects are replicated and made available across multiple nodes so that the offered load for each object, which we refer to as the \emph{demand} for the object, can be split across \emph{multiple nodes (service choices)}.
The best support for load balancing is achieved when each object is stored at each node, but that is feasible only in exceptional cases at large scale.
If the demand for each object is known and fixed, each object could be stored with an adequate level of redundancy. However, in practice, object popularities, and in turn their demands, are not only unknown but also fluctuate over time. Thus, load balancing should be robust against skews and changes in object popularities \cite{ChallengesInBuildingLargeScaleInformationRetrievalSystems:Dean09, Scarlett:AnanthanarayananAK11}.




Load balancing has been considered in two important settings.
In the first, we call the \emph{dynamic setting}, load balancing is addressed from the point of view of scheduling tasks for processing.
Here the nodes correspond to single queues which are processed independently and in parallel.
Load balancing amounts to interrogating some subset of all the queues and offering new tasks to those nodes which are least loaded.
The simplest of these models is the one in which each task is sent to the node with the least number of tasks. Clearly this achieves the ideal load balance.
This scheme is only practical for a relatively small number of nodes but becomes unworkable for large-scale systems with tens of thousands of nodes or more.
For this reason, a great deal of attention has been placed on developing schemes which offer tasks to a restricted number of nodes. These schemes include those based on the well-known \emph{power of $d$ choices} paradigm.
A range of asymptotic results have been obtained following this direction, often using analysis based on balls into bins models \cite{BallsIntoBins:RaabS98, BalancedAllocations:AzarBK99, BalancedAllocations_HeavilyLoadedCase:Berenbrink20}. 

All of the above literature addresses the load balancing question from the point of view of spreading tasks evenly across the nodes.
Its main weakness however is that the well-understood power of $d$ choices is only applicable for systems where the arrivals can be placed at any one of the nodes. For instance, for scheduling compute tasks across nodes within the same data center.
However, the flexibility of querying any $d$ bins at random does not exist in storage systems. This is because each object is typically stored only at a limited number of nodes and an arriving request can only be served at one of the nodes that host the requested object.

A more appropriate model for storage is to suppose that each request is offered to one of a subset of nodes, each of which hosts the requested object.
Kenthapadi and Panigrahy propose a model in \cite{BalancedAllocationsOnGraphs:Kenthapadi06} along these lines where subsets of nodes are represented as edges in a graph. They studied this restricted model for the power of two choices with $n$ balls and $n$ bins. Edges are selected according to incoming object requests and then the arrival is assigned to the least loaded node among the vertices of the edge.
Godfrey then extended this in \cite{BalancedAllocationsOnHypergraphs:Godfrey08} to general power of $d$ choices.
Applying this model to a system with $n$ storage nodes with each object being replicated $d = \Omega\left(\log(n)\right)$ times, Godfrey's results lead to the conclusion that effective load balancing can be achieved in the sense of power of $d$ choices. 
Godfrey then goes on to show that if $d$ grows more slowly, then the above conclusion for the power of $d$ schemes is no longer valid.
Storage schemes we consider in this paper are a natural special case of the balanced allocations on hyper graphs that was considered by Godfrey.
Beyond the above conclusions, Godfrey's results provide little insight into practical storage schemes. For example, how to distribute different objects across various storage nodes or what gain can be made from using coded schemes based on object XOR's.
Finally Godfrey's results are shown only for the lightly loaded case, which is when the cumulative load offered on the system scales as the order of the number of nodes.
This paper extends Godfrey's results for the case with concrete storage schemes and without restricting ourselves to the lightly loaded case.
Furthermore we examine
i) the number of different objects stored per node,
ii) object overlaps between the storage nodes,
iii) using coded objects rather than plain replicas, and address their impact on storage efficiency and load balancing performance.


Under the dynamic setting, load offered for the objects is not known a priori. Requests arrive sequentially and each is assigned to a node based on the current load at each node.
We now turn to the second setting, namely the \emph{static setting}.
In this setting, a different question is asked: is it feasible to carry the load, if the load offered for each object is known from the start.
Any assignment strategy realized under the dynamic setting is also achievable under the static setting.
This is because knowing the offered loads for the stored objects in advance makes it only easier to balance the load.
Load balancing performance in the static setting therefore represents the best-case performance of the system.

The question asked in the static setting gives rise to two distinct approaches.
The first approach leads to the design of redundancy schemes, namely \emph{batch codes}. They balance the load as long as any $m$ objects are chosen with replacement out of all objects and then requested simultaneously \cite{BatchCodesAndTheirApps:IshaiKO04}.
The storage schemes we consider fall into the class of combinatorial batch codes \cite{CombinatorialBatchCodes:StinsonWP09}.
We should note that the static model adopted for batch codes has been extended to a more dynamic setting, which led to the design of \emph{asynchronous} batch codes \cite{AsyncBatchCodes:RietST}. Batch codes were originally designed to balance only a single batch of requests. An asynchronous variant is designed to balance the present batch together with the upcoming batch or batches.
This approach asks the question the other way around and seeks to find the set of all object demand vectors that can be supported by a system with a given storage scheme, namely the system's \emph{service capacity region} \cite{AllertonServiceCapacity:AktasJS17, ITWServiceCapacity:AndersonJJ18}.
Our treatment of load balancing falls into this second approach.
References \cite{AllertonServiceCapacity:AktasJS17} and \cite{ITWServiceCapacity:AndersonJJ18} only address the case where each node stores a single object, their approach being to find the system's complete service capacity region. However determining this region with multiple objects at each node appears to be a largely intractable problem.
In our approach, we rely on a new stochastic formulation which allows us to analyze the load balancing in this scenario and to draw a range of conclusions on the design and structure of storage schemes.
Furthermore, the primary goal in the service rate approach is to evaluate system stability under different object popularities. In this paper however we address the related problem of \emph{feasibility} of load balancing, i.e., the degree to which storage resources can be adapted according to the changes of the object popularities.

It is helpful at this point to add a few words by way of explaining the model setup that we use in this paper.
First, we consider only \emph{regular balanced} storage schemes in which each object is replicated $d$ times (hence regular) and each node stores the same number of objects (hence balanced).
Storage schemes specify where each object copy is stored, and therefore determine the set of all possible ways that one can split and assign individual object demands across the nodes.
Second, as far as object demands are concerned, we suppose that the cumulative load is fixed and that all object popularities are equally likely.
This is motivated by the fact that the cumulative demand for all the objects stored in the system is known to vary slowly over time and therefore is easy to estimate (see, e.g., Fig.~7 in \cite{StudyOfMapreduceWorkloads:ChenAK12}). Individual demands for objects fluctuate much more rapidly.
Additionally, our assumption of fixed cumulative load on the system is the continuous generalization of the offered load model used in the batch code problem.

We now turn to the metrics which we will be using to analyze load balancing performance.
These metrics can be understood by considering Fig.~\ref{fig:pyramid_coverage}.
Fig.~\ref{fig:pyramid_coverage} shows a simplex region that corresponds to the set of all possible demand vectors for three objects.
If load balancing is ideal, then for all these vectors stability can be achieved. However, this is too onerous in practice as it would require an unacceptably large storage overhead.
A compromise therefore is to minimize the fraction of demand vectors which cannot be supported.
Under our formulation, this corresponds to our assumption that the object demands are uniformly distributed on the simplex region. This is indeed the standard model used in the study of load balancing in the dynamic setting.
Overall, we measure the \emph{robustness} of load balancing as the probability $P_{\Sigma}$ that the system will be stable when the demand vector is sampled uniformly at random from the simplex region defined by cumulative load $\Sigma$.
We also use another metric that is closely connected to $P_{\Sigma}$ to measure the load imbalance.
Precisely for a system of $n$ nodes under a cumulative load of $\Sigma$, the load imbalance, $\mathcal{I}$, is given by minimizing the maximum load and dividing it by its minimum possible value $\Sigma/n$.

\begin{figure}[t]
  \centering
  \includegraphics[width=.3\textwidth]{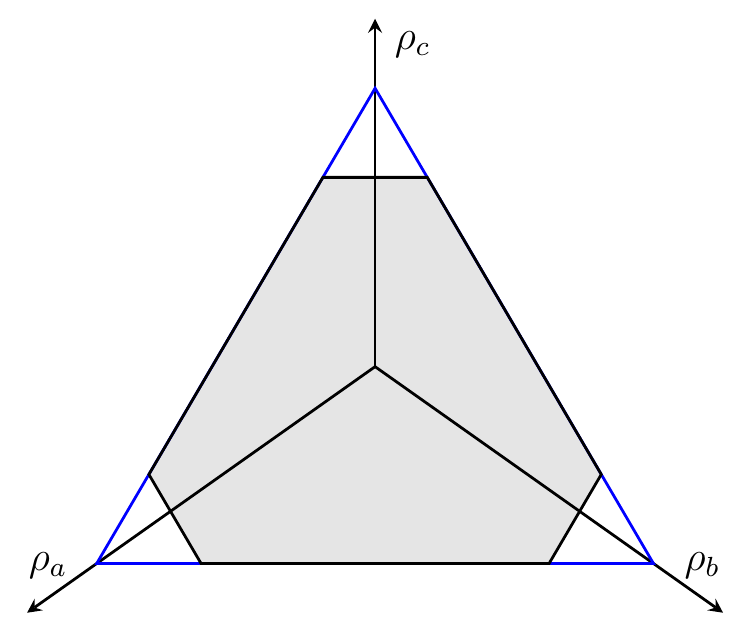}
  \caption{An illustration of the demand offered for three objects $a$, $b$, $c$ and the fraction of those that are supported by the system.
  The simplex region corresponds to all demand vectors that can be offered on the system. It is given by $\{(x, y, z) \mid x + y + z = \Sigma\}$ for cumulative load $\Sigma$.
  The shaded region shows the demand vectors that are supported by the system.}
\label{fig:pyramid_coverage}
\end{figure}

With the metrics now established, we can state the main contributions of this paper.
These can be understood by considering the following questions with respect to load balancing and storage:
\begin{enumerate}[label=\textbf{Q\arabic*}, leftmargin=*]
\item \label{q:1} How does $\mathcal{I}$ scale with the number of objects and nodes in the system where there is no storage redundancy?
\item \label{q:2} Does the degree of overlap between the service choices of different objects play a critical role in terms of achieving better load balance? How does $\mathcal{I}$ depend on the number of service choices $d$ provided for each object?
\item \label{q:3} XOR'ing reduces storage requirements, however can effective load balancing still be achieved using XOR's, rather than object replicas?
\end{enumerate}
We address these questions for common storage schemes with regular balanced redundancy. Optimizing storage schemes for various purposes are studied elsewhere (e.g., for improving data access in \cite{allocations:SardariRFS10}).

\vspace{1ex}
\noindent
\textbf{Our contribution}:\space
From the results of the paper, we are able to conclude the following answers \textbf{A1}, \textbf{A2} and \textbf{A3} for the questions \textbf{Q1}, \textbf{Q2} and \textbf{Q3} above.

\vspace{1ex}
\noindent
\textbf{A1}:\space
For the storage schemes with no redundancy, we find $\mathcal{I} = \Theta\left(\log(n)/m\right)$ where $m$ is the number of different objects stored on each node. This implies that in the limit as $n \to \infty$, load imbalance grows as $\log(n)$.
This implies that if we want to maintain a load of $\Sigma/n$ in the maximally loaded node, we need $n\log(n)$ nodes in the system.

\vspace{1ex}
\noindent
\textbf{A2}:\space
To answer \ref{q:2}, we consider $d$-replication storage schemes.
Different storage schemes under a regular balanced requirement lead to differences in the way objects overlap at the nodes.
By examining three different storage designs, we find that the scheme with consistently small overlaps outperforms schemes with fewer overlaps which are necessarily larger in size.
For one class of schemes with limited overlap, which we call the r-gap schemes, we have obtained the following asymptotic results
$\mathcal{I} = \Theta\left(\log(n)/d\right)$ when $d = o\left(\log(n)\right)$ and $\mathcal{I} = \Theta\left(\log\log(n)/\log(n)\right)$ when $d = \Theta\left(\log(n)\right)$.
These results imply that
i) creating $d$ service choices for each object initially reduces the load imbalance in the system multiplicatively by $d$,
ii) there is an exponential reduction in load imbalance as soon as $d$ reaches of order $\log(n)$.
This quantifies the tradeoff between storage and service capacity for r-gap schemes.



\vspace{1ex}
\noindent
\textbf{A3}:\space
Using XOR's reduces the amount of storage needed but at the same time increases the capacity needed to access various objects. This is because to obtain a single object, we have to access several object codes. An $r$-XOR object code is one which is constructed from $r$ objects.
Our asymptotic results show that storage with $d$-fold redundancy implemented with $r$-XOR's have the advantage that it has the same scaling of $\mathcal{I}$ in $d$ \emph{as if} the service choices were created with replicas.
Thus in large-scale systems, there is no loss of significant benefit over replication. XOR'ing can be used to trade off between storage and the access capacity.

The paper is organized as follows:
Sec.~\ref{sec:intro} gives an overview of the literature on load balancing in storage context. We also discuss the connections between our approach and the prior work.
Sec.~\ref{sec:preliminaries} presents our storage and offered load model and its connection with uniform spacings. We also precisely define the metrics that we use to evaluate load balancing performance.
In Sec.~\ref{sec:1choice} we consider storage schemes with no redundancy and answer \ref{q:1}.
In Sec.~\ref{sec:dchoice} we consider storage schemes with object replication and answer \ref{q:2}.
In Sec.~\ref{sec:dchoice_wxors} we consider creating storage redundancy with XOR's rather than object replicas and answer \ref{q:3}.

\section{System Model and Performance Metrics}
\label{sec:preliminaries}
In this section, we introduce our system model and define the metrics we use to evaluate load balancing performance.
We study load balancing in the static setting with a continuous service and offered load model.
Our model reveals an interesting connection of the load balancing problem to convex polytopes and the spacings between ordered random uniform samples, the so-called \emph{uniform spacings} \cite{Spacings:Pyke65}. We elaborate on this in Sec.~\ref{subsec:offered_load}.
The latter connection enabled us to apply prior results on uniform spacings in answering the questions posed in Sec.~\ref{sec:intro}.

\subsection{Storage and Access Model}
\label{subsec:sys_model}
We consider a system of $n$ storage nodes $s_1, \dots, s_n$ hosting $k$ data objects $o_1, \dots, o_k$, possibly with redundancy. 
Each node provides the same \emph{capacity} for content access, which is defined as the maximum number of bytes that can be streamed from a node per unit time.
An \emph{object} denotes the smallest unit of content, and mathematically, it is a fixed-length string of bits. XOR'ing multiple objects is carried out bitwise.

We refer to the \emph{offered load} for object $o_i$ as its \emph{demand} $\rho_i$.
Demand for an object represents the average number of bytes streamed from the system per unit time to access the object, divided by a single node's content access capacity.
We refer to a node that hosts an object as a \emph{service choice} for the object.
Multiple service choices for an object can be created by replicating it over several nodes. We consider $d$-choice storage schemes with replicas in Sec.~\ref{sec:dchoice}.
Alternatively, XOR'ed object copies can be used to create multiple service choices. We consider $d$-choice storage schemes with XOR's in Sec.~\ref{sec:dchoice_wxors}.
When XOR'ing is used, a service choice for an object refers to a \emph{recovery set}, that is, a set of nodes that can jointly recover the object.
Accessing an object through one choice should not interfere with accessing the same object through another choice. Different service choices for the same object are therefore disjoint.

Demand for an object can be arbitrarily split across its service choices. When a load of $\rho$ is exerted by an object on a recovery set, each node within the set will be offered a load of $\rho$.
The \emph{load on a node} is given by the sum of the offered load portions exerted on it by the objects for which the node can serve as a choice.
A node is said to be \emph{stable} if the load on it is less than $1$. A system is said to be stable if every node within the system is stable.
We assume that each of the object demands $\rho_i$ is split across its service choices so that the load on the maximally loaded node is minimized. As we describe further in Sec.~\ref{subsec:cap_region}, this can be obtained by solving a norm minimization problem given the storage scheme and the value of $\rho_i$.

A \emph{storage allocation} defines how each object is assigned, possibly with redundancy, to storage nodes.
This paper focuses on \emph{regular balanced $d$-choice} storage allocations.
\begin{definition}
  A \underline{regular balanced $d$-choice allocation} stores each object with $d$ service choices and distributes object copies across the nodes so that each node stores the same number of different objects (either as an exact or XOR'ed copy).
\label{def:reg_balanced_dchoice_alloc}
\end{definition}
There are many ways to design a $d$-choice allocation. We detail some of them in Sec.~\ref{sec:dchoice} and \ref{sec:dchoice_wxors}. In the rest of the paper, unless otherwise noted, the allocation itself will refer to a regular balanced allocation.

\vspace{1ex}
\noindent
\textbf{Connection with batch codes}:
A $(k, N, m, n, t)$ \emph{batch code} encodes $k$ objects into $N$ copies with redundancy and distributes them across the $n$ nodes in such a way that any $m$ of these objects can be accessed by reading at most $t$ objects from any node \cite{BatchCodesAndTheirApps:IshaiKO04}.
The goal while designing batch codes is to minimize the total storage requirement.
Redundancy can be either in the form of replicating individual objects or encoding (e.g., XOR'ing) multiple objects together.
Multiset batch codes are concerned with a more general case in which the selection of $m$ objects for access is done with replacement. It should be noted that each of the objects needs to be accessed separately, that is, the content that is read for accessing an object cannot be used to access another object.
We should note that the demand model assumed for batch codes have been extended to cases with additional constraints, such as balancing the access frequencies over the nodes \cite{AccessBalancingInDistributedStorage:DauM18}. This area of research has investigated the use of certain class of redundancy schemes to balance access when the popularity ranks of the objects are known in advance.

In Sec.~\ref{sec:dchoice} we will consider storage allocations that are constructed by object replication.
Such allocations implement batch codes as follows.
\begin{lemma}
  Any $d$-choice regular balance storage allocation with object replication
  represents a $(k, kd, n, n, 1)$ batch code and a $(k, kd, d, n, 1)$ multiset batch code.
\label{lm:dchoice_alloc_is_batchcode}
\end{lemma}
\begin{proof}
  See Appendix~\ref{subsec:proof_lm_dchoice_alloc_is_batchcode}.
\end{proof}
Batch codes with replication are known as \emph{combinatorial} batch codes and their construction has been well studied \cite{CombinatorialBatchCodes:StinsonWP09, CombinatorialDesigns:Stinson07}. In particular, a combinatorial batch code is named as \emph{$d$-uniform} if it stores each object in exactly $d$ nodes, which is exactly the $d$-choice requirement we consider here. An approach that is based on block design has been given in \cite{OptimalBatchCodesBasedOnBlockDesigns:SilbersteinG16} to construct optimal $d$-uniform batch codes.

\subsection{Offered Load and Uniform Spacing Model}
\label{subsec:offered_load}
We suppose that the system can be offered any object demand (offered load) vector $(\rho_1, \dots, \rho_k)$ in the set
\begin{equation}
  \mathcal{S}_{\Sigma} =
  \Bigl\{(\rho_1, \dots, \rho_k)\;\Bigl|\; \sum_{i=1}^k \rho_i = \Sigma, ~\rho_i \geq 0\Bigr\}.
\label{eq:S_Sigma}
\end{equation}
That is, the cumulative offered load remains constant but the object popularities can change arbitrarily.
The cumulative condition we impose on the offered load is the continuous generalization of the load model assumed in the multiset batch code problem. Recall that a multiset batch code is designed to support a user who can simultaneously access $m$ objects that are selected with replacement out of all objects stored in the system.

We further assume that the demand vector $(\rho_1, \dots, \rho_k)$ is sampled uniformly at random from $\mathcal{S}_{\Sigma}$.
Uniform distribution across all possible demand vectors models the case where no a priori knowledge is given on the object popularities. In other words, it represents the case with maximum uncertainty about the object popularities. 
This assumption is the continuous generalization of what has been used in balls-into-bins models. There each ball arrives for one of the stored objects chosen uniformly at random from all stored in the system.
In addition as we discuss shortly, sampling demand vectors uniformly at random is able to model the skewed nature of object popularities in real systems \cite{Scarlett:AnanthanarayananAK11}. However it should be noted that modeling with a more general distribution would yield additional insight on load balancing under more specific and possibly more realistic offered load models.
An example of such a model would be one that puts larger probability mass on the demand vectors representing skewed object popularities.


In what follows, we define uniform spacings. They are mathematical objects connected with the uniform sampling of points from a simplex.
Let $U_{(1)}, \dots, U_{(k-1)}$ be $k-1$ i.i.d.\ uniform samples in $[0, 1]$, given in non-decreasing order. Then 
$S_i = U_{(i)} - U_{(i-1)}$ for $i = 1, \dots, k$, where $U_{(0)} = 0$ and $U_{(k)} = 1$, are known as $k$ uniform spacings on the unit line.
\begin{lemma}[see e.g. \cite{Spacings:Pyke65}]
  Uniform spacings $(S_1, \dots, S_k)$ are uniformly distributed over the simplex
  \[ \Bigl\{(s_1, \dots, s_k)\; \Bigl|\; \sum_{i=1}^k s_i = 1, ~s_i \geq 0 ~\mathrm{for}~ i=1, \dots, k\Bigr\}. \]
\label{lm:spacings_uniformlydisted}
\end{lemma}
Lemma~\ref{lm:spacings_uniformlydisted} implies that object demands $\rho_i$ in our model under a cumulative load of $\Sigma$ can be seen as $k$ uniform spacings in $[0, \Sigma]$.
This connection allows us to use the results on uniform spacings to evaluate load balancing performance in systems with $d$-choice storage allocation. We do the evaluation in terms of the performance metrics defined in the following subsection.

We next examine the popularity skew characteristics captured by our uniform offered load model (as promised above).
Without loss of generality, let us assume that the cumulative demand $\Sigma$ offered on the system is $1$.
Let $N(\alpha, \beta)$ denote the number of objects with a demand of $\geq \alpha$ and $\leq \beta$.
Then $N(\alpha, \beta)$ is given by the number of uniform spacings that are within $[\alpha, \beta]$.
An asymptotic characterization of $N(\alpha, \beta)$ has been given in \cite{OnSumOfFunctionsOfUniformSpacings:Darling53} as follows.
\begin{theorem}(\cite[Theorem 8.1-2-3]{OnSumOfFunctionsOfUniformSpacings:Darling53})
  \begin{enumerate}[label=\emph{R\arabic*}, leftmargin=*]
  \item \label{N_a_b_medium}. $N(\alpha/k, \beta/k)$ is asymptotically normally distributed as $k \to \infty$ with an asymptotic mean and variance
    \begin{equation*}
    \begin{split}
        &\mu_k \sim k\left(e^{-\alpha} - e^{-\beta}\right), \\
        &\sigma_k^2 \sim k\left(e^{-\alpha} - e^{-\beta} - (\alpha e^{-\alpha} - \beta e^{-\beta})^2 \right).
    \end{split}
  \end{equation*}
  
  \item \label{N_a_b_small}. $N(\alpha/k^2, \beta/k^2)$ has an asymptotic Poisson distribution with parameter $\beta-\alpha$.
  
  \item \label{N_a_b_large}. $N((\log(k) + \alpha)/k, (\log(k) + \beta)/k)$ has an asymptotic Poisson distribution with parameter $e^{-\alpha} - e^{-\beta}$.
  \end{enumerate}
\label{thm:N_a_b}
\end{theorem}
Results in Theorem~\ref{thm:N_a_b} tell us a great deal about the object popularities implemented by our demand model.
With high probability, only a few of the objects will be highly popular ($\rho \sim \log(k)/k$), only a few will have very low popularity ($\rho \sim 1/k^2$), while most objects will have around-average popularity ($\rho \sim 1/k$).
This reflects the skewed object popularities observed in real storage systems (see e.g. Fig.~3 in \cite{FacebookDataCaching:HuangBV13}).

\subsection{Storage Service Capacity}
\label{subsec:cap_region}
We now obtain mathematical expressions determining the \emph{service capacity region} for a storage system.
In particular, we will express the set of all object demand vectors under which the system with a given storage allocation can operate under stability.
Service capacity for systems that store content with erasure coding was first studied in \cite{AllertonServiceCapacity:AktasJS17} and further studied in \cite{ITWServiceCapacity:AndersonJJ18}.
We adopt a formulation similar to the one introduced in \cite{AllertonServiceCapacity:AktasJS17}.
The formulation we present in this section provides a geometric interpretation of the performance metrics $\mathcal{P}_{\Sigma}$ and $\mathcal{I}$ introduced in Sec.~\ref{subsec:perf_metrics}.
\begin{definition}
  \underline{Service capacity region} for a system with a given storage allocation is the set of all object demand vectors $\bm{\rho} = (\rho_1, \dots, \rho_k)$ under which the system can operate under stability.
 \label{def:serv_cap}
\end{definition}
\noindent

In what follows, we explain how to express the service capacity region as a solution for a system of linear inequalities. 
Let us consider a system in which object $o_i$ is stored on $d_i$ nodes for $i = 1, \dots, k$. Then its demand $\rho_i$ can be distributed across its $d_i$ service choices, each handling a fraction of $\rho_i$.
Let us denote the portion of $\rho_i$ that is assigned to the $j$th choice of $o_i$ with $\rho_i^{(j)}$. Then we have $\rho_i = \rho_i^{(1)} + \dots + \rho_i^{(d_i)}$.
We represent the stacked collection of all these per-node demand portions with the following vector of length $d_1 + \dots + d_k$:
\[ \bm{x}^T = \left(\rho_1^{(1)}, \dots, \rho_1^{(d_1)}, ~\ldots~, \rho_k^{(1)}, \dots, \rho_k^{(d_k)} \right). \]
Converting back to $\bm{\rho}$ from $\bm{x}$ is a matter of matrix-vector multiplication as $\bm{\rho} = \bm{T} \cdot \bm{x}$, where $\bm{T}$ is a binary matrix of size $k \times (d_1 + \dots + d_k)$.
System stability is ensured if and only if the total demand flowing into each node is less than its capacity $1$. This can be expressed as a linear inequality for each node and a matrix inequality for the whole system of $n$ nodes as
\begin{equation}
  \bm{M} \cdot \bm{x} \prec \bm{1}, \quad \bm{x} \succeq \bm{0},
\label{eq:Mx_leq_C}
\end{equation}
where $\prec$ and $\succeq$ denote the standard partial orderings in $R^n$, and $\bm{0}$ and $\bm{1}$ denote the all-zeros and ones vectors of length $n$, respectively.
The overall service capacity region of the system is given by
\begin{equation}
  \mathcal{C} = \left\{\bm{\rho} ~\mid~ \exists \bm{x}; ~\bm{M} \cdot \bm{x} \prec \bm{1}, ~\bm{T} \cdot \bm{x} = \bm{\rho}, ~\bm{x} \succeq \bm{0} \right\}.
\label{eq:serv_cap}
\end{equation}

$\bm{M}$ expresses the storage allocation and it is a binary matrix of size $n \times (d_1 + \dots + d_k)$.
It is constructed by setting $\bm{M}[i, j]$ to $1$ if the demand portion $\bm{x}[j]$ flows into node-$i$, and to $0$ otherwise.
When storage redundancy is created with only object replicas, each column of $\bm{M}$ becomes a binary representation of a node that stores the corresponding object copy. Precisely, each column of $\bm{M}$ would consist of a single $1$, and the position of this $1$ within the column is equal to the position of the represented node within the sequence of all nodes.
For instance for the system that stores $a$, $b$, $c$ across three nodes by allocating two service choices for each as $\left\{(a, c), ~(b, a), ~(c, b)\right\}$, we have
\begin{equation*}
    \begin{split}
        &\bm{x}^\intercal = \left(\rho_a^{(1)}, \rho_a^{(2)}, \rho_b^{(1)}, \rho_b^{(2)}, \rho_c^{(1)}, \rho_c^{(2)} \right), \\
        &\bm{M} = \begin{bmatrix}
                 1 & 0 & 0 & 0 & 0 & 1 \\
                 0 & 1 & 1 & 0 & 0 & 0 \\
                0 & 0 & 0 & 1 & 1 & 0
                \end{bmatrix}.
    \end{split}
\end{equation*}

When storage redundancy consists of coded objects, some of the demand portions $\rho_i^{(j)}$ might be assigned to recovery sets. 
A recovery set for an object is a set of nodes from which the object can be recovered. 
When a demand portion of $\rho$ is assigned to a recovery set, then a fraction of $\rho$ capacity will be used up at each node within the recovery set.
Then the columns of $\bm{M}$ that consist of multiple ones represent recovery sets for the corresponding objects.
For instance, for the storage allocation $\left\{(a, b+c), ~(b, a+c), ~(c, a+b)\right\}$, we have
\begin{equation*}
    \begin{split}
        &\bm{x}^\intercal = \left(\rho_a^{(1)}, \rho_a^{(2)}, \rho_b^{(1)}, \rho_b^{(2)}, \rho_c^{(1)}, \rho_c^{(2)} \right), \\
        &\bm{M} = \begin{bmatrix}
                 1 & 0 & 0 & 1 & 0 & 1 \\
                 0 & 1 & 1 & 0 & 0 & 1 \\
                0 & 1 & 0 & 1 & 1 & 0
                \end{bmatrix}.
    \end{split}
\end{equation*}

\begin{lemma}
  The service capacity region for any storage system is a convex polytope.
\label{lm:cap_region_is_convex}
\end{lemma}
\begin{proof}
  The convex polytope expressed by \eqref{eq:Mx_leq_C} in $R_+^{d_1 + \dots + d_k}$ consists of all demand portion vectors $\bm{x}$ under which the system is stable. Capacity region $\mathcal{C}$ is the linear transformation of this polytope by $\bm{T}$. Hence $\mathcal{C}$ is another convex polytope in $R_+^k$.
\end{proof}

As noted in Sec.~\ref{subsec:sys_model}, we consider the case where object demands $\rho_i$ are split across their choices such that the load on the maximally loaded storage node is minimized.
This means that for a given object demand vector $\bm{\rho}$, out of all demand portion vectors $\bm{x}$ that satisfy \eqref{eq:Mx_leq_C}, the system will split the demands across the nodes according to $\bm{x}^*$. This achieves the best possible load balance.
Thus $\bm{x}^*$ is the optimal solution for the following convex optimization problem:
\begin{equation}
  \min_{\bm{x}} ~  \norm{\bm{M} \cdot \bm{x}}_{\infty}; \quad \bm{T} \cdot \bm{x} = \bm{\rho}, ~\bm{x} \succeq \bm{0},
\label{eq:min_prob_for_loadbalancing}
\end{equation}
where $\norm{\cdot}_{\infty}$ denotes the infinity norm.

Copying an object to a node that did not previously host it, increments the number of service choices for the object.
We next state a simple but useful fact as the first step to understanding the gains of increasing the number of service choices for the objects.

\begin{lemma}
  Let the system capacity region be $\mathcal{C}$ for a given storage allocation.
  Keeping the number of nodes fixed, let us store an object replica (or a coded copy) on a node that did not previously host the object (or any object present in the coded copy). Let $\mathcal{C}^\prime$ be the system capacity region for this modified allocation.
  Then $\mathcal{C} \subset \mathcal{C}^\prime$.
\label{lm:addingchoice_expands_C}
\end{lemma}
\begin{proof}
  See Appendix~\ref{subsec:proof_lm_addingchoice_expands_C}.
\end{proof}

\subsection{Performance Metrics}
\label{subsec:perf_metrics}
We now give precise definitions for the two metrics that we use to quantify load balancing performance in distributed storage.
The first metric measures the system's \emph{robustness} against the presence of skews and changes in object popularities. 
We quantify robustness as the fraction of demand vectors that are supported by the system in the simplex that consists of all vectors that sum up to $\Sigma$.
\begin{definition}[Measure of robustness]
  For a system with a given storage allocation, let the capacity region be the polytope $\mathcal{C}$ and let $\mathcal{S}_{\Sigma}$ be defined for a given cumulative load $\Sigma$ as in \eqref{eq:S_Sigma}.
  \underline{$\mathcal{P}_\Sigma$} for the system is given by
  \begin{equation}
    \mathcal{P}_{\Sigma} = \frac{\mathrm{Volume}\left(\mathcal{C} \cap \mathcal{S}_{\Sigma}\right)}{\mathrm{Volume}\left(\mathcal{S}_{\Sigma}\right)}.
  \label{eq:P_Sigma}
  \end{equation}
\label{def:P_Sigma}
\end{definition}
$\mathcal{P}_\Sigma$ is obviously $0$ when $\Sigma > n$, hence we assume $\Sigma \leq n$ implicitly throughout.
The shaded region in Fig.~\ref{fig:pyramid_coverage} illustrates the intersection of the simplex $\mathcal{S}_{\Sigma}$ and the system capacity region.
Recall that the demand vector $(\rho_1, \dots, \rho_k)$ offered on the system is sampled uniformly at random from $\mathcal{S}_{\Sigma}$.
Therefore another way to define $\mathcal{P}_{\Sigma}$ is that it is the probability that the system defined by $\mathcal{S}_\Sigma$ will be stable. In other words, $\mathcal{P}_{\Sigma}$ is the \emph{probability of robustness} for a system that operates under a cumulative demand of $\Sigma$.

The expression given for $\mathcal{P}_{\Sigma}$ in \eqref{eq:P_Sigma} is a useful geometric interpretation.
It implies that once the capacity region of a system is determined, evaluating $\mathcal{P}_{\Sigma}$ for it becomes a computational geometry problem.
Finding volumes or pairwise intersections of convex polytopes are well studied problems, and numerous efficient algorithms are available to compute both in the literature, e.g., see  \cite{ExactVolumeComputationForPolytopes:Bueler00}.
Eq.\eqref{eq:P_Sigma} essentially gives a recipe to exactly compute $\mathcal{P}_{\Sigma}$ for a system with any given storage allocation.
This, together with the fact that service capacity region is a convex polytope (Lemma~\ref{lm:cap_region_is_convex}), implies $\mathcal{P}_{\Sigma}$ is non-increasing in $\Sigma$.
\begin{corollary}
  For any system, if $\Sigma > \Sigma^\prime$ then $\mathcal{P}_{\Sigma} \leq \mathcal{P}_{\Sigma^\prime}$.
\label{cor:P_Sigmaless_geq_P_Sigma}
\end{corollary}
\begin{proof}
  See Appendix~\ref{subsec:proof_cor_P_Sigmaless_geq_P_Sigma}.
\end{proof}

The second metric measures the load imbalance across the storage nodes.
In the balls-into-bins model, load imbalance is quantified by the number of balls in the maximally loaded bin.
Our metric is a continuous generalization of this. In addition, we relate the load on the maximally loaded node to its smallest possible value. This makes our metric independent of the cumulative offered load.
\begin{definition}[Measure of load imbalance]
  Consider the system with $n$ storage nodes operating under a cumulative load of $\Sigma$.
  Load imbalance factor $\mathcal{I}$ for the system is defined as minimum of the maximal load on any node over all feasible loads, divided by its minimum possible value, i.e., $\Sigma/n$.
  
  It is given as
  \begin{equation}
    \mathcal{I} = \frac{\norm{\bm{M} \cdot \bm{x}^*}_{\infty}}{\Sigma/n}.
  \label{eq:I}
  \end{equation}
  where $\bm{M}$ is the binary matrix representing the system's storage allocation (as described in Sec.~\ref{subsec:cap_region}), and $\bm{x}^*$ is the solution for the minimization program given in \eqref{eq:min_prob_for_loadbalancing}.
\label{def:I}
\end{definition}
Notice that $\mathcal{I}$ is always $\geq 1$.
We abstract away the resource sharing dynamics and other system related aspects and evaluate performance through the metrics of load imbalance.
It should be noted that this is the same approach taken in previous studies with balls-into-bins models.
Even though much complexity is abstracted away from the system model, load balance is a good proxy to get an understanding of system's performance in terms of response time.
This is because the more evenly the load is balanced, the smaller the system's response time is.
We argue this for illustrative purposes as follows. Let us suppose that resource sharing at each storage node is implemented with a first-come first-served (FCFS) or a processor sharing (PS) queue. Response time at a node will then get increasingly larger and more variable the greater the load is to the node.
Let us denote the average response time of node $i$ under an average offered load of $\rho_i \in (0, 1)$ with $T(\rho_i)$.
We know that under either FCFS or PS queue, $T(\rho)$ will scale as $\rho/(1 - \rho)$, which is a convex monotonically increasing function of $\rho$.
Roughly, $\rho_i/\sum_{i=1}^n$ fraction of the request arrivals will experience an average response time of $T(\rho_i)$.
By Jensen's inequality, we find for all non-negative $\rho_i$ that
\[ \sum_{i=1}^n \frac{\rho_i}{\sum_{i=1}^n \rho_i} T(\rho_i) \geq T\left(\frac{1}{n} \sum_{i=1}^n \rho_i\right). \]
Notice that the right hand side of the above inequality is the average system response time when the load is perfectly balanced across the nodes.
This tells us that perfect load balance minimizes the average response time under any cumulative offered load.
The same argument given above can be used to observe that load balance is desirable to optimize other convex quantities such as the second moment of response time.

Lemma~\ref{lm:addingchoice_expands_C} given in the previous sub-section says that storing an additional redundant object copy in the system expands the service capacity region that is supported by the system. Keep in mind that the total capacity in the system does not change but, by creating additional storage redundancy, we are able to use the available capacity more efficiently for content access.
This helps to show how increasing storage redundancy improves load balancing performance in terms of $\mathcal{P}_{\Sigma}$ and $\mathcal{I}$.
\begin{corollary}
  Consider a system with load balancing performance of $\mathcal{P}_{\Sigma}$ and $\mathcal{I}$.
  Suppose a new redundant object copy is stored in the system as described in Lemma~\ref{lm:addingchoice_expands_C}, and let $\mathcal{P}_{\Sigma}^\prime$, $\mathcal{I}^\prime$ represent the metrics of load balancing performance for the modified system.
  Then we have $\mathcal{P}_{\Sigma}^\prime \geq \mathcal{P}_{\Sigma}$ and $\mathcal{I}^\prime \leq \mathcal{I}$.
\end{corollary}
\begin{proof}[Proof sketch]
  By Lemma~\ref{lm:addingchoice_expands_C}, $\mathcal{C} \subset \mathcal{C}^\prime$.
  This together with \eqref{eq:P_Sigma} directly implies $\mathcal{P}_{\Sigma}^\prime \geq \mathcal{P}_{\Sigma}$.
  In order to see $\mathcal{I} \geq \mathcal{I}^\prime$, it is enough to observe the following.
  For an arbitrary object demand vector $\bm{\rho}$, let $\bm{x}$ be the demand portion vector that minimizes the load on the maximally loaded node in the unmodified system.
  Given that $\mathcal{C} \subset \mathcal{C}^\prime$, $\bm{x}$ is also achievable by the modified system.
\end{proof}

\subsection{Note on the Proofs and the Notation}
We place the proofs in the Appendix, in order not to disrupt continuity of the text.
Throughout the paper, $\log$ refers to the natural logarithm, and $\log_{(i)}$ refers to $i$ times iterated logarithm, e.g., $\log_{(2)}(x)$ stands for $\log\log(x)$.
We denote convergence of a sequence using the ``$\to$'' notation.
Suppose that $\{f_n(x); n \geq 1\}$ is a sequence of functions $f_n : D \to \mathbb{R}$ and $f : D \to \mathbb{R}$.
If $\lim_{n \to \infty} f_n(x) = f(x)$ at every $x \in D$, we will denote this as $f_n(x) \to f(x)$.
Throughout this paper, we will rely on almost sure convergence in probability.
Suppose that $X_1, X_2, \ldots$ is a sequence of random variables in a sample space $\Omega$ and suppose $X$ is another random variable in $\Omega$.
Then $\{X_n; n \geq 1\}$ converges to $X$ almost surely if
\[ \Prob{\omega \in \Omega : \lim_{n \to \infty} X_n(w) = X(w)} = 1.\]
We denote almost sure convergence of $X_n$ to $X$ as $X_n \to X$ a.s.

\section{Load Balancing with no Redundancy}
\label{sec:1choice}
In this section we consider \emph{single-choice} allocations in which each of the $k$ objects is stored on only a single node and each of the $n$ nodes stores $m=k/n$ different objects. We assume $n|k$.
Demand for each object in this case has to be completely served by the only node hosting the object, and each node has to serve the total demand for all objects stored on it.

As discussed in Sec.~\ref{subsec:offered_load}, object demand vector $(\rho_1, \dots, \rho_k)$ can be described by $k$ uniform spacings in $[0, \Sigma]$.
Given that uniform spacings are exchangeable RV's, we can say without loss of generality that if node $s_i$ stores objects $o_{(i-1)m + 1}, \dots, o_{im}$, then load $l_i$ exerted on $s_i$ is given by $l_i = \sum_{(i-1)m + 1}^{im} \rho_j$.
For the system to be stable, all $l_i$ must be $< 1$.
Thus $\max\{l_1, \dots, l_n\} < 1$ is necessary and sufficient for system stability.
This implies $\mathcal{P}_\Sigma$ for the system is given by $\Prob{\max\{l_1, \dots, l_n\} < 1}$.
In the uniform spacing literature, random variables (RV's) $l_i$ have been studied and are called non-overlapping $m$-spacings. Their maximum is referred to as the \emph{maximal non-overlapping $m$-spacing}.

\begin{definition}
  \underline{Maximal non-overlapping $m$-spacing} for $k$ uniform spacings on the unit line is defined for $k = m \cdot n$ as
    \begin{equation*}
        M^{(\text{no})}_{k, m} =\! \max_{i = 1, \dots, n} U_{(im)} - U_{((i-1)m)},
    \end{equation*}
or as
    \begin{equation*}
        \max_{i = 1, \dots, n} \!\sum_{j=(i-1)m+1}^{im} S_{j}.
    \end{equation*}
\label{def:Mnonoverlapping}
\end{definition}

In a system of $n$ nodes storing $k$ objects under a cumulative demand of $1$, load on the maximally loaded node is given by $M^{(\text{no})}_{k, m}$.
Using a combination of the ideas presented in \cite{AsymptoticsForMaxOverSum:Darling52, EntropyAndMaximalSpacingsForRandomPartitions:Slud78, OrderStatisticsOfUniformSpacings:Devroye81}, we can derive the following convergence results for $M^{(\text{no})}_{k, m}$.
\begin{lemma}
  For fixed $m$, as $n \to \infty$
  \begin{equation}
    \Prob{M^{(\text{no})}_{k, m} \cdot m n - \log(n) - f_n < x} \overset{d}{\to} G(x). 
  \label{eq:Mnonoverlappingd_convergence_indist}
  \end{equation}
  where $G(x) = \exp\left(-\exp(-x)\right)$ is the standard Gumbel function and $f_n = (m-1)\log_{(2)}(n) - \log((m-1)!)$.
  
  In the limit $n \to \infty$, the following inequality holds
  \begin{equation}
    M^{(\text{no})}_{k, m} \leq \frac{\log(n)}{m n} + O\left(\frac{\log_{(2)}(n)}{n}\right) \quad \text{a.s.}
  \label{eq:Mnonoverlappingd_convergence_as_werr}
  \end{equation}
  Furthermore, as $n \to \infty$
  \begin{equation}
    \frac{M^{(\text{no})}_{k, m} \cdot m n}{\log(n)} \to 1 \;\text{a.s.}
  \label{eq:Mnonoverlappingd_convergence_as}
  \end{equation}
\label{lm:Mnonoverlappingd_convergence_as}
\end{lemma}
\begin{proof}
  See Appendix~\ref{subsec:proof_lm_Mnonoverlappingd_convergence_as}.
\end{proof}

Now we are ready to express the metrics $\mathcal{P}_\Sigma$ and $\mathcal{I}$ for a system with single-choice allocation. 
\begin{lemma}
  In a system with a single-choice storage allocation,
  \begin{equation}
    \mathcal{P}_\Sigma = \Pr\left\{M^{(\text{no})}_{k, m} < 1/\Sigma\right\}, \qquad \mathcal{I} = M^{(\text{no})}_{k, m} \cdot n.
  \label{eq:P_I_d1_exact}
  \end{equation}
\label{lm:P_I_d1}
\end{lemma}
\begin{proof}
  When system operates under a cumulative offered load of $\Sigma$, the load on the maximally loaded node is given by $M^{(\text{no})}_{k, m} \cdot \Sigma$.
  This together with the definition of $\mathcal{P}_\Sigma$ (Def.~\ref{def:P_Sigma}) and $\mathcal{I}$ (Def.~\ref{def:I}) gives us \eqref{eq:P_I_d1_exact}.
\end{proof}

Using Lemma~\ref{lm:P_I_d1} and Lemma~\ref{lm:Mnonoverlappingd_convergence_as}, we determine the behavior of $\mathcal{P}_\Sigma$ and $\mathcal{I}$ for large $n$ as follows:
\begin{theorem}
  Consider a system with single-choice storage allocation.
  For fixed $m$, as $n \to \infty$
  \begin{equation}
    \Prob{\mathcal{I} \cdot m - \log(n) - f_n < x} \to G(x),
  \label{eq:I_d1_convergence_indist}
  \end{equation}
  \begin{equation}
    \frac{\mathcal{I} \cdot m - f_n}{\log(n)} \to 1 \quad\text{a.s.}
  \label{eq:I_d1_convergence_as}
  \end{equation}
  where $f_n = (m-1)\log_{(2)}(n) - \log((m-1)!)$.
  
  If $\Sigma_n = b_n \cdot n/\log(n)$ for some sequence $b_n > 0$, then as $n \to \infty$
  \begin{equation}
    \mathcal{P}_{\Sigma_n} \to \begin{cases}
      1 & \limsup b_n < m, \\
      0 & \liminf b_n > m.
    \end{cases}
  \label{eq:P_d1_convergence}
  \end{equation}
\label{thm:P_I_d1}
\end{theorem}
\begin{proof}
  See Appendix~\ref{subsec:proof_thm_P_I_d1}
\end{proof}

\begin{remark}
  Theorem~\ref{thm:P_I_d1} implies for a system of $n$ nodes with no redundancy and fixed $m$ that load imbalance $\mathcal{I} = \Theta\left(\log(n)\right)$.
  This is due to the following fact.
  As $n$ increases, maximal load on any node decays as $\Sigma/n$ in the case with perfect load balancing.
  However due to the skews in popularity, demands for the popular objects go down as $\log(n)/n$ (recall Theorem~\ref{thm:N_a_b}).
  This is why the load imbalance in the system grows with $n$ as $\log(n)$.
  The scaling of load imbalance with $\log(n)$ as we find here is aligned with the well-known result derived in the dynamic load balancing setting:
  if $n$ balls arrive sequentially and each is placed into one of the $n$ bins randomly, the maximally loaded bin will end up with $\Theta(\log(n)/\log_{(2)}(n))$ balls w.h.p. \cite{BalancedAllocations:AzarBK99}.
  
  In addition \eqref{eq:I_d1_convergence_as} shows that $\mathcal{I}$ decays multiplicatively with $m$ for large scale systems.
  This implies that in large scale systems, if the system can support a cumulative load of $\Sigma$ while storing $k$ objects, it will also be able to support a cumulative load of $r \times \Sigma$ while storing $k \times r$ objects.
\label{rm:1choice}
\end{remark}

\section{Load Balancing with \texorpdfstring{$d$}{TEXT}-fold Redundancy}
\label{sec:dchoice}
In this section we consider \emph{$d$-choice} allocations in which each of the $k$ objects is stored on $d$ different nodes ($d$-choices) and each of the $n$ nodes stores $k d/n$ different objects.

As discussed in Remark~\ref{rm:1choice}, load imbalance $\mathcal{I}$ in the system decays with the number of objects ($m$) stored per node.
We here, and also in Sec.~\ref{sec:dchoice_wxors}, consider the worst case for load balancing, that is $k = n$. This makes the problem more tractable to formulate and study the load balancing problem. This also makes it easier to explain and interpret the derived results.
Results that we present now can be extended for the general case with a fixed value of $k/n > 1$ by using arguments that are very similar to those we discuss.

In what follows, we first define and discuss \emph{maximal $d$-spacing}.
It is a mathematical object defined in terms of uniform spacings and is instrumental while deriving our results.
We then present our results on the load balancing performance of systems with $d$-choice allocations.

\subsection{Uniform Spacings Interlude}
\label{subsec:uniform_spacings_interlude}
\begin{definition}[$M_{k, d}$]
  Maximal $d$-spacing within $k$ uniform spacings on the unit line is defined as
  \[ M_{k, d} = \max_{i = 0, \dots, k-d} U_{(i+d)} - U_{(i)} \quad\text{or}\quad \max_{i = 1, \dots, k-d+1} \sum_{j=i}^{i+d-1} S_{j}, \]
  where $d$ is any integer in $[1, k]$, and $U_0=0$ and $U_{k}=1$ as given in Sec.~\ref{subsec:offered_load}.
\label{def:Md}
\end{definition}

It is worth to note that maximal $d$-spacing $M_{k, d}$ defined above is the overlapping counterpart of the maximal non-overlapping $m$-spacing $M^{(\text{no})}_{k, m}$ defined in Def.~\ref{def:Mnonoverlapping}.
We extensively use the results presented in \cite{EntropyAndMaximalSpacingsForRandomPartitions:Slud78, StrongLawsForMaximalKSpacing:DeheuvelsD84, UniformExpSpacings:Devroye86, WeakLimitOfMaximalUniformKspacing:MijatovicV16} on $M_{k, d}$ while deriving our main results.
We state the ones that we use in the remainder.




\begin{lemma}(\cite[Theorem 1]{WeakLimitOfMaximalUniformKspacing:MijatovicV16})
  For any integer $d \geq 1$, as $k \to \infty$
  \begin{equation}
   \begin{split}
     \Pr\Bigl\{& M_{k, d} \cdot k - \log(k) \\
     &~ - (d-1)\log_{(2)}(k) - \log((d-1)!) \leq x\Bigr\} \to G(x).
   \end{split}
  \label{eq:Md_convergence_indist}
  \end{equation}
\label{lm:Md_convergence_indist}
\end{lemma}

\begin{lemma}(\cite[Theorem 2, 6]{StrongLawsForMaximalKSpacing:DeheuvelsD84})
  For $d = o(\log(k))$, as $k \to \infty$
  \begin{equation}
    \frac{M_{k, d} \cdot k - \log(k)}{(d-1)\left(1 + \log_{(2)}(k) - \log(d)\right)} \to 1 ~~\text{a.s.}
  \label{eq:Md_convergence_as_smalld}
  \end{equation}
  
  For $d = c\log(k) + o(\log_{(2)}(k))$ with some constant $c > 0$,
  \[ V_k = \frac{M_{k, d} \cdot k - (1 + \alpha)c \cdot \log(k)}{\log_{(2)}(k)} \]
  satisfies
  \begin{equation}
  \begin{split}
     \limsup\limits_{k \to \infty} V_k &= \beta^\ast (1 + \alpha)/\alpha ~~\text{a.s.} \\
     \liminf\limits_{k \to \infty} V_k &= -\beta^\dagger (1 + \alpha)/\alpha ~~\text{a.s.}
  \end{split}
  \label{eq:Md_convergence_as_largerd}
  \end{equation}
  where $\alpha$ is the unique positive solution of $e^{-1/c} = (1 + \alpha)e^{-\alpha}$, and $\beta^\ast$ and $\beta^\dagger$ are constants taking values in $[-0.5, 1.5]$ and $[-1.5, -0.5]$ respectively.
\label{lm:Md_convergence_as}
\end{lemma}

The uniform spacings we have considered so far are defined on the unit line segment.
However, in order to evaluate load balancing performance in systems with $d$-choice storage allocation, we need to consider uniform spacings on the \emph{unit circle}.
\begin{definition}[$M_{k, d}^{(c)}$]
  Maximal $d$-spacing within $k$ uniform spacings on the unit circle is defined as
  \[ M_{k, d}^{(c)} = \max_{i = 1, \dots, k} \sum_{j=i}^{i+d-1} S_{i}, \quad \text{where } S_{i} = S_{i - k} \text{ for }i > k. \]
\label{def:Mcirculard}
\end{definition}
We show in Appendix~\ref{subsec:proof_Mcirculard} that results stated in Lemma~\ref{lm:Md_convergence_indist} and \ref{lm:Md_convergence_as} for maximal $d$-spacing on the unit line carry over to its counterpart defined on the unit circle.


\subsection{Evaluating \texorpdfstring{$\mathcal{P}_\Sigma$}{TEXT} and \texorpdfstring{$\mathcal{I}$}{TEXT} for \texorpdfstring{$d$}{TEXT}-choice Storage Allocation}
\label{subsec:dchoice_alloc_eval}
A $d$-choice storage allocation defines a $d$-regular balanced bipartite mapping from the set of objects to the set of nodes, which we refer to as the \emph{allocation graph}.
Its construction can be described as follows:
i) Map primary copies for all objects to nodes with a bijection $f_0$,
ii) For $i$ going from $1$ to $d$,
map the $i$th redundant object copies to nodes with a bijection $f_i$ such that $f_i(o) \neq f_j(o)$ for every $j < i$ and $o$.
Thus every node stores a single primary and $d-1$ redundant object copies, and each copy stored on the same node is for a different object. We refer to a node with the index of the primary object copy stored on it, i.e., object $o_i$ is hosted primarily on $s_i$. The subscript in $s_i$ or $o_i$ will implicitly denote $i \mod n$ throughout.
We denote the set of nodes that host object $o_i$ with $C_i$. In other words $C_i$ consists of the service choices available for $o_i$.

The number of service choices available for the objects is not the only factor that impacts load balancing performance in storage systems.
Layout of the content across the storage nodes also plays a role in load balancing.
Given that the total number of object copies stored in the system is greater than the number of storage nodes, we need to have $|C_i \cap C_j| > 0$ for some $j \neq i$.
Overlaps between $C_i$ might lead to contention when the content popularity is skewed towards objects with overlapping service choices.
Both the number of overlapping $C_i$ and the size of the overlaps should be minimized in order to improve load balancing performance.
However, size and number of overlaps cannot be reduced together given a fixed number of nodes in the system.
For every regular balanced $d$-choice allocation with object replicas we have
\begin{equation}
  \sum_{i=1}^k \sum_{j \neq i} \left|C_i \cap C_j \right| = (d-1)d \cdot k,
\label{eq:S_Sigmaum_of_intersection_cardinalities}
\end{equation}
where $|\cdot|$ denotes the cardinality of a set.
This equality follows by observing that each node serves as a service choice for $d$ different objects and each is counted in exactly $d-1$ of the overlaps $C_i \cap C_j$.
It should be noted that the sum in \eqref{eq:S_Sigmaum_of_intersection_cardinalities} is equal to \emph{twice} the cumulative cardinality of the overlaps between all pairs of $(C_i, C_j)$.
We emphasize that cumulative cardinality of the pairwise overlaps is fixed. Hence reducing the size of overlaps between some of the sets $C_i$ can only come at the cost of enlarging overlaps between some other $C_j$.

We first consider the two following simple designs for constructing $d$-choice storage allocations.

\vspace{1ex}
\noindent
\textbf{Clustering design}:\space
This design is possible only if $d|n$.
Let us partition the nodes into $n/d$ sets, each of which we call a \emph{cluster}.
Let us then assign each object to a cluster such that each cluster is assigned exactly $d$ objects.
Every object is stored across all the nodes within its assigned cluster.
The resulting storage has an allocation graph that is composed by $n/d$ disjoint $d$-regular complete bi-partite graphs.
Hence we obtain a storage scheme with $d$-choice allocation.
For instance, 3-choice allocation for 9 objects $a, \dots, i$ with cyclic design would look like
\begin{equation*}
  \icol{a\\b\\c} ~~\icol{a\\b\\c} ~~\icol{a\\b\\c} ~~\icol{d\\e\\f} ~~\icol{d\\e\\f} ~~\icol{d\\e\\f} ~~\icol{g\\h\\i} ~~\icol{g\\h\\i} ~~\icol{g\\h\\i}.
\end{equation*}

\vspace{1ex}
\noindent
\textbf{Cyclic design}:\space
In this design we follow a \emph{cyclic} construction.
We start by assigning the original object copies to the nodes according to an arbitrary bijection $f_0$.
We assign the $i$th replicas of the objects for $i = 1, \dots, d-1$ by using bijection $f_i$, where $f_i$ is obtained by applying circular shift on $f_0$ repeatedly $i$ times. Note that shifting is applied in the same direction in all the steps.
In other words, we pick $f_i$ such that $f_{i+1}(o) = f_i(o) + 1 \mod n$ for $i = 0, \dots, d-1$ and every $o$.
For instance, 3-choice allocation for 7 objects $a, \dots, g$ with cyclic design would look like
\begin{equation}
  \icol{a\\g\\f} ~~\icol{b\\a\\g} ~~\icol{c\\b\\a} ~~\icol{d\\c\\b} ~~\icol{e\\d\\c} ~~\icol{f\\e\\d} ~~\icol{g\\f\\e}.
\label{eq:3choice_cyclic}
\end{equation}

For a given set of objects $S$, the union of their choices $C_i$ forms the \emph{node expansion} of $S$, which we denote by $N(S)$.
If $|N(S)| = x$, then there is at most $x$ amount of capacity available for the joint use of the objects within $S$. It is surely impossible to stabilize the system when the cumulative demand for $S$ is greater than $N(S)$.
Thus it is desirable to increase the size of the node expansions in the allocation graph in order to guarantee stability for larger skews in content popularity.
Greater expansion for a given $S$ requires reducing the size of the overlaps between the choices $C_i$ for the objects in $S$. This would imply overlapping the choices $C_i$ of objects within $S$ with the choices $C_j$ of objects outside of $S$.

It is not easy to define a knob that regulates both the overlaps between object choices $C_i$ and the node expansions in the allocation graph.
We next define a class of allocations in which the overlaps and node expansions are loosely controlled by a single parameter.
\begin{definition}[$r$-gap design]
  An allocation is an $r$-gap design if $|C_i \cap C_j| = 0$ for $j > i$ and $\min\{j-i, n-(j-i)\} > r$.
\label{def:rgap_design}
\end{definition}

\begin{lemma}
  In a $d$-choice allocation with $r$-gap design, $r \geq d-1$
  and $x \leq |N(S)| \leq x+2r$ for any $S = \{o_i, o_{i+1}, \dots, o_{i+x-1}\}$ for any $i$. 
\label{lm:on_rgap}
\end{lemma}
\begin{proof}
  See Appendix~\ref{subsec:proof_lm_on_rgap}.
\end{proof}

We can use the properties of $r$-gap design to find necessary and sufficient conditions for the stability of a storage system.
\begin{lemma}
  Consider a system with $d$-choice storage allocation that is constructed with an $r$-gap design and operating under a cumulative demand of $\Sigma$.
  Then for system stability, a necessary condition is given as
  \begin{equation}
    M^{(c)}_{n, i} \leq (i + 2r)/\Sigma, \quad\text{for any } i = 1, \dots, n-2r,
   \label{eq:rgap_necccond_for_stability}
  \end{equation}
  and a sufficient condition is given as
  \begin{equation}
    M^{(c)}_{n, r+1} \leq d/\Sigma.
   \label{eq:rgap_suffcond_for_stability}
  \end{equation}
  
  In other words, we have for $i = 1, \dots, n-2r$
  \[ \Prob{M^{(c)}_{n, r+1} \leq d/\Sigma} \leq \mathcal{P}_\Sigma \leq \Prob{M^{(c)}_{n, i} \leq (i + 2r)/\Sigma}. \]
\label{lm:rgap_neccsuffcond_for_stability}
\end{lemma}
\begin{proof}
  See Appendix~\ref{subsec:proof_lm_rgap_neccsuffcond_for_stability}.
\end{proof}


Notice that clustering or cyclic design is an $r$-gap design.
Hence the necessary and sufficient conditions given in Lemma~\ref{lm:rgap_neccsuffcond_for_stability} for system stability are valid for storage allocation with either design.
In addition, the well-defined structure of these two designs allows us to refine the results given in Lemma~\ref{lm:rgap_neccsuffcond_for_stability} as follows.
\begin{lemma}
  In a $d$-choice allocation constructed with clustering or cyclic design
  \begin{equation*}
    \Prob{M^{(c)}_{n, d} \leq d/\Sigma} \leq \mathcal{P}_\Sigma \leq \Prob{M^{(c)}_{n, d+1} \leq 2d/\Sigma}.
  \end{equation*}
\label{lm:clustering_cyclic_neccsuffcond_for_stability}
\end{lemma}
\begin{proof}
  See Appendix~\ref{subsec:proof_lm_clustering_cyclic_neccsuffcond_for_stability}.
\end{proof}

Using the bounds given in Lemma~\ref{lm:clustering_cyclic_neccsuffcond_for_stability}, we can find an asymptotic characterization for $\mathcal{P}_\Sigma$ and $\mathcal{I}$ as follows.
\begin{theorem}
  Consider a system with $d$-choice storage allocation constructed with clustering or cyclic design.
  
  When $d = o\left(\log(n)\right)$, the following inequality hold in the limit as $n \to \infty$
  \begin{equation}
    \frac{1}{2} \leq \frac{\mathcal{I} \cdot d}{\log(n) + (d-1)(1 + \log_{(2)}(n) - \log(d))} \leq 1 \quad \text{a.s.}
  \label{eq:I_dchoice_w_clustering_or_cyclic_convergence_as_smalld}
  \end{equation}
  If $\Sigma_n = b_n \cdot n/\log(n)$ for some sequence $b_n > 0$, then as $n \to \infty$
  \begin{equation}
    \mathcal{P}_{\Sigma_n} \to \begin{cases}
      1 & \limsup b_n/d < 1, \\
      0 & \liminf b_n/2d > 1.
    \end{cases}
  \label{eq:P_dchoice_w_clustering_cyclic_convergence_smalld}
  \end{equation}
  
  When $d = c\log(n)$ for some constant $c > 0$, the following inequality holds in the limit $n \to \infty$
  \begin{equation}
    \frac{1}{6} \leq \frac{2c \alpha}{3(\alpha+1)} \cdot \frac{\mathcal{I} \cdot \log(n)}{\log_{(2)}(n)} \leq 1 \quad \text{a.s.},
  \label{eq:I_dchoice_w_clustering_or_cyclic_convergence_as_largerd}
  \end{equation}
  where $\alpha$ is the unique positive solution of $e^{-1/c} = (1 + \alpha)e^{-\alpha}$.
  If $\Sigma_n = b_n \cdot n/\log(n)$ for some sequence $b_n > 0$, then as $n \to \infty$
  \begin{equation}
      \mathcal{P}_{\Sigma_n} \to \begin{cases}
      1 & \limsup b_n \cdot 1.5\tau/d < 1, \\
      0 & \liminf b_n \cdot 0.25\tau/d > 1,
    \end{cases}
  \label{eq:P_dchoice_w_clustering_cyclic_convergence_largerd}
  \end{equation}
  where $\tau = c(1 + \alpha)^2/\alpha$.
\label{thm:P_I_dchoice_w_clustering_cyclic}
\end{theorem}
\begin{proof}
  See Appendix~\ref{subsec:proof_thm_P_I_dchoice_w_clustering_cyclic}
\end{proof}

\begin{remark}
  Theorem~\ref{thm:P_I_dchoice_w_clustering_cyclic} implies that load imbalance
  $\mathcal{I} = \Theta\left(\log(n)/d\right)$ when $d = o\left(\log(n)\right)$, and $\mathcal{I} = \Theta\left(\log\log(n)/\log(n)\right)$ when $d = \Theta\left(\log(n)\right)$.
  These imply that
  i) $d$ choices for each object initially reduces load imbalance multiplicatively by $d$,
  ii) there is an exponential reduction in load imbalance as soon as $d$ reaches of order $\log(n)$, that is, what we started with $\mathcal{I} = \Theta(\log(n))$, for $d = 1$, is exponentially larger than what we end up with $\mathcal{I} = \Theta\left(\log\log(n)/\log(n)\right)$ for $d = \Theta\left(\log(n)\right)$.
  The second implication extends Godfrey's first observation in \cite{BalancedAllocationsOnHypergraphs:Godfrey08} to the static setting under general offered load. Godfrey derived his results with the dynamic balls-into-bins model under light offered load, i.e., when $O(n)$ balls are sequentially placed into $n$ bins.
  Godfrey's second observation implies in the dynamic setting that for any $1 < d < \Theta\left(\log(n)/\log\log(n)\right)$, there exists a $d$-choice storage allocation such that incrementing $d$ from $1$ to $2$ will not yield exponential reduction in load imbalance.
  The first implication given above, shows this observation of Godfrey for a concrete storage design and extends it to the static setting under general offered load.
\label{rm:dchoice}
\end{remark}

In what follows, we discuss some of our arguments using simulation results on $\mathcal{I}$.
We compute $\mathcal{I}$ by taking an average of its values obtained from $10^{5}$ simulation runs.
Within each simulation run, the object demand vectors that are offered to the system are sampled uniformly at random from the simplex $\mathcal{S}_\Sigma$, which is defined by a fixed $\Sigma$ as in \eqref{eq:S_Sigma}.

Fig.~\ref{fig:I_wd} plots $\mathcal{I}$ for $n=100$.
Notice that $\mathcal{I}$ is close to $\log(n)$ when $d=1$ as suggested by Theorem~\ref{thm:P_I_d1}.
As $d$ is incremented, $\mathcal{I}$ decays as $1/d$ as suggested by Theorem~\ref{thm:P_I_dchoice_w_clustering_cyclic}.
This illustrates that our asymptotic results are close estimates for the finite case.
\begin{figure}[t]
  \centering
  \begin{subfigure}
    \centering
    \includegraphics[width=.4\textwidth]{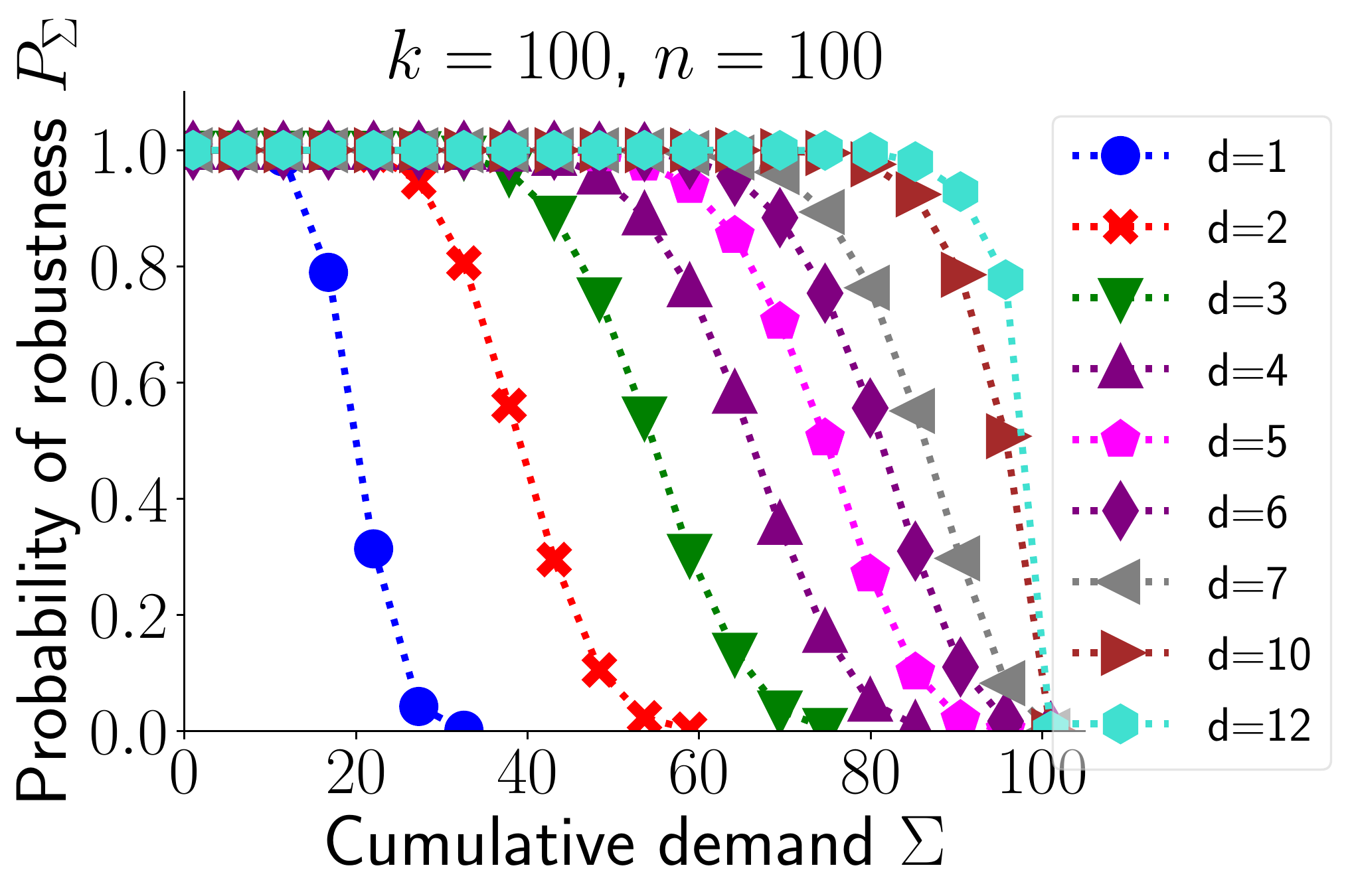}
  \end{subfigure}
  \begin{subfigure}
    \centering
    \includegraphics[width=.4\textwidth]{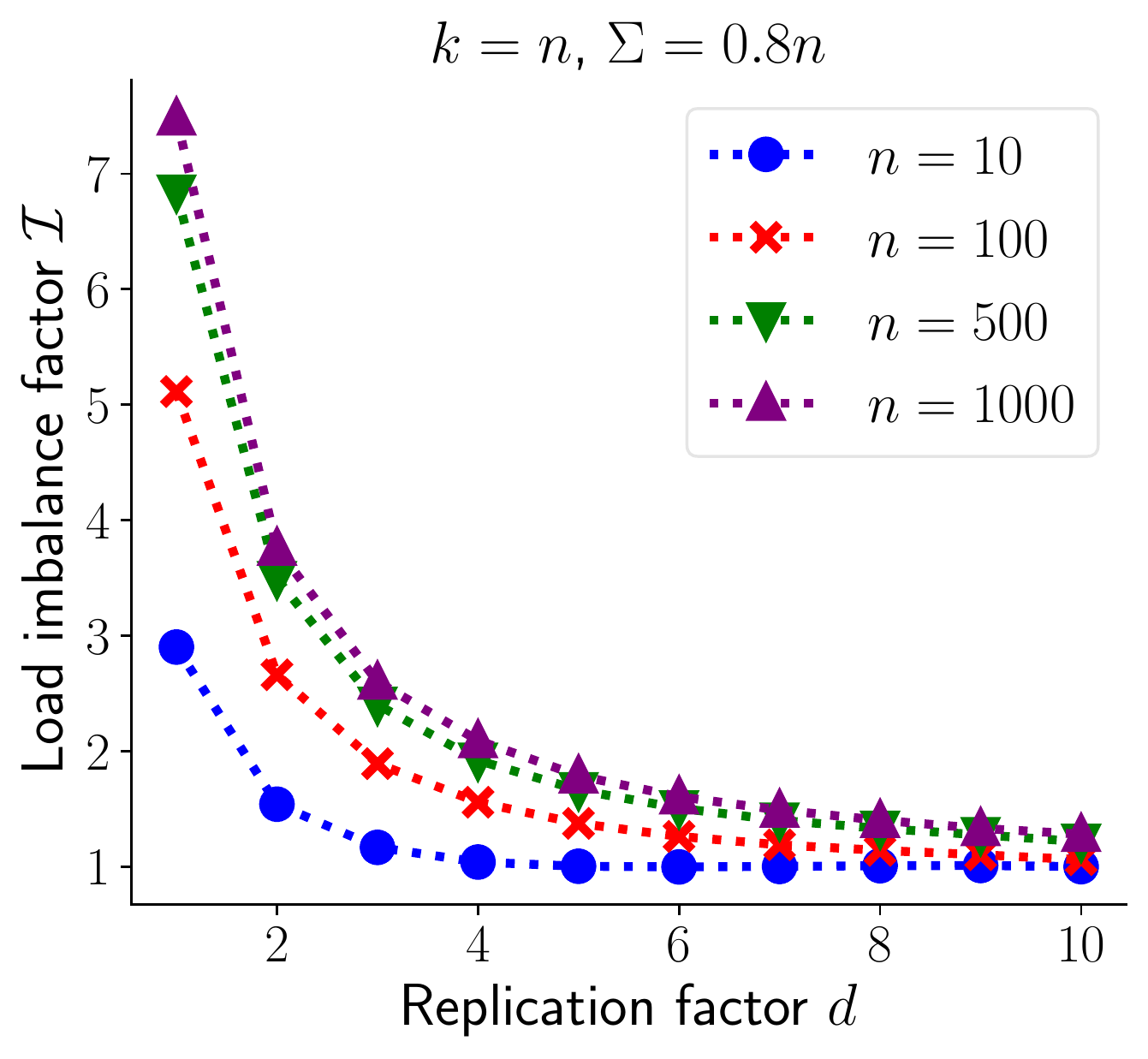}
  \end{subfigure}
  \caption{Simulated average values for $\mathcal{P}_{\Sigma}$ and $\mathcal{I}$ for a system that implements $d$-choice allocation with cyclic design.
  Cumulative offered load on the system is set to $\Sigma=0.8n$.
  System stores $k=n$ objects across $n$ nodes.}
\label{fig:I_wd}
\end{figure}

Constructions with $r$-gap designs decouple object choices $C_i$ that are $r$-apart at the cost of enlarging the overlaps between those that are close to each other, as in the clustering or cyclic designs.
The Balanced Incomplete Block Designs (BIBD) allow control of the overlaps between every pair of $C_i$.

\begin{definition}[BIBD, \cite{CombinatorialDesigns:Stinson07}]
  A $(d, \lambda)$ block design is a class of equal-size subsets of $\mathcal{X}$ (the set of stored objects), called blocks (storage nodes), such that every point in $\mathcal{X}$ appears in exactly $d$ blocks (service choices), and \emph{every} pair of distinct points is contained in exactly $\lambda$ blocks.
\label{def:block_design}
\end{definition}

Since we assume $k = n$, the block designs we consider are symmetric.
A symmetric BIBD with $\lambda=1$ guarantees that $|C_i \cap C_j| = 1$ for every $j \neq i$.
Since this case represents the minimal overlap between sets $C_i$ we focus on this case. The block design we consider refers to $(d, 1)$ symmetric BIBD, that is every object appears in $d$ nodes and every pair of distinct objects is contained in exactly one node.
Since every pair of $C_i$ overlaps at one node, we have
\[ \sum_{i=1}^k \sum_{j \neq i} \left|C_i \cap C_j \right| = (k-1)k. \]
Then by \eqref{eq:S_Sigmaum_of_intersection_cardinalities}, such block designs are possible only if $k = d^2 - d + 1$.
For instance a 3-choice allocation with a block design is given as
\begin{equation}
  \icol{a\\b\\c} ~~\icol{a\\f\\g} ~~\icol{a\\d\\e} ~~\icol{b\\d\\f} ~~\icol{b\\e\\g} ~~\icol{c\\d\\g} ~~\icol{c\\e\\f}
\label{eq:eq_3choice_bibd}
\end{equation}

The sufficient and necessary conditions presented in Lemma~\ref{lm:rgap_neccsuffcond_for_stability} cannot be used on a storage allocation with a block design, since they are not $r$-gap designs.
However using ideas that are similar to those used to derive Lemma~\ref{lm:rgap_neccsuffcond_for_stability}, we can find the following conditions for system stability.
\begin{lemma}
  Consider a system with $d$-choice allocation constructed with a block design and operating under a cumulative offered load of $\Sigma$.
  For stability of the system, a necessary condition is given as $M^{(c)}_{n, d} \leq (d^2 - 2d + 3)/\Sigma$
  and a sufficient condition is given as $M^{(c)}_{n, d} \leq d/2\Sigma$.
  
\label{lm:bibd_neccsuffcond_for_stability}
\end{lemma}
\begin{proof}
  See Appendix~\ref{subsec:proof_lm_bibd_neccsuffcond_for_stability}.
\end{proof}

The stability conditions given in Lemma~\ref{lm:bibd_neccsuffcond_for_stability} allow us to find bounds on $\mathcal{P}_\Sigma$ and $\mathcal{I}$ for storage allocations with block designs, similar to those that were stated in Theorem~\ref{thm:P_I_dchoice_w_clustering_cyclic}.
We do not state them here since they are obtained by simply modifying the multiplicative factors in the bounds given in Theorem~\ref{thm:P_I_dchoice_w_clustering_cyclic}.
The upper bound on $\mathcal{I}$ in this case decays as $1/d$ with increasing $d$, which says that providing $d$ service choices for each object initially reduces load imbalance at least multiplicatively by $d$. However, the lower bound on $\mathcal{I}$ decays in this case as $1/d^2$, that is, block designs can possibly implement better scaling of $\mathcal{I}$ in $d$ compared to clustering or cyclic designs.

Our asymptotic analysis does not allow ordering different designs of $d$-choice storage allocations in terms of their load balancing performance.
As discussed previously, all $d$-choice allocations yield the same cumulative overlap between object choices $C_i$ (recall \eqref{eq:S_Sigmaum_of_intersection_cardinalities}) and each design gives a different way of distributing the overlaps across object choices $C_i$.
With simulations we find that it is better to evenly spread the overlaps between choices $C_i$ using a block design, that is, \emph{many but consistently small overlaps are better than fewer but occasionally large overlaps}.
Fig.~\ref{fig:fig_I_cl_vs_cyclic_vs_bibd} shows the average $\mathcal{I}$ for systems with 3- and 5-choice allocations that are constructed using clustering, cyclic or block designs.
We see here, and in other simulation results we omit, that the largest gain in load balancing is achieved by moving from clustering to cyclic, while moving further to block design yields a smaller gain in $\mathcal{I}$.
Furthermore cyclic designs exist for any value of $k$, $n$ and $d < n$, while block designs exist only for a restricted set of $k$, $n$ and $d$.
Cyclic design therefore appears to be favorable for constructing multiple-choice storage allocations in real systems.

Currently we don't have a rigorous way to understand how designs with different overlaps compare with each other in terms of $\mathcal{P}_\Sigma$ or $\mathcal{I}$.
In the following subsection, we present our intuitive reasoning on why consistently small overlaps is better in terms of load balancing than shrinking overlaps between some objects and making them larger for others.

\begin{figure}[h]
  \centering
  \begin{subfigure}
    \centering
    \includegraphics[width=0.3\textwidth]{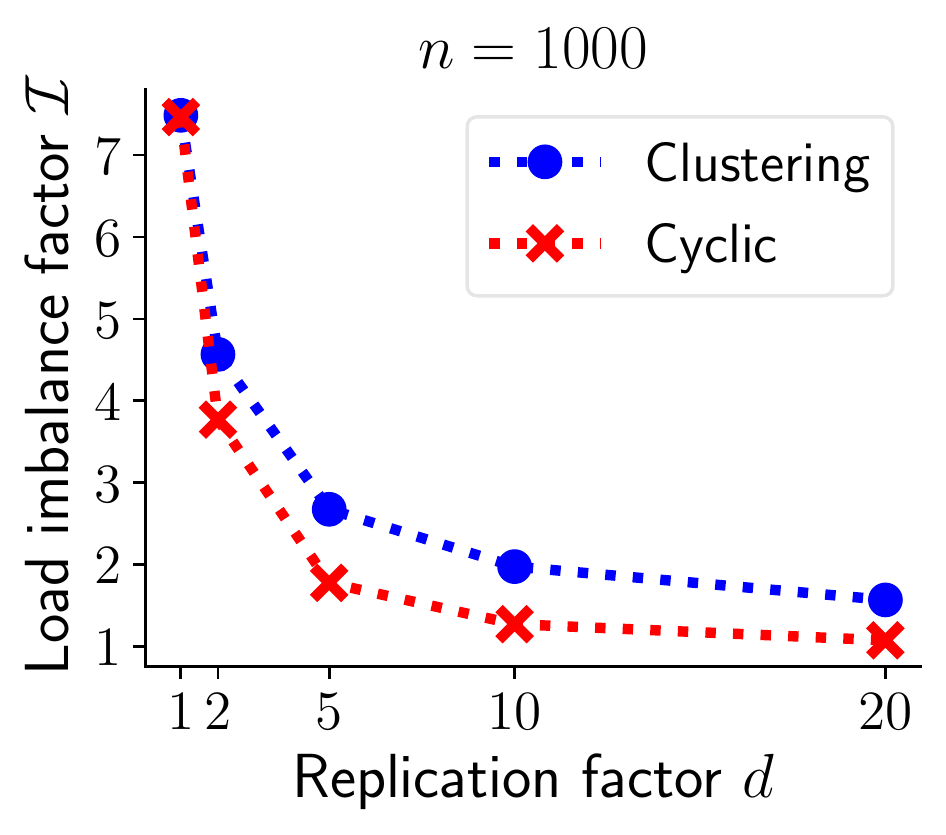}
  \end{subfigure}
  \begin{subfigure}
    \centering
    \includegraphics[width=0.3\textwidth]{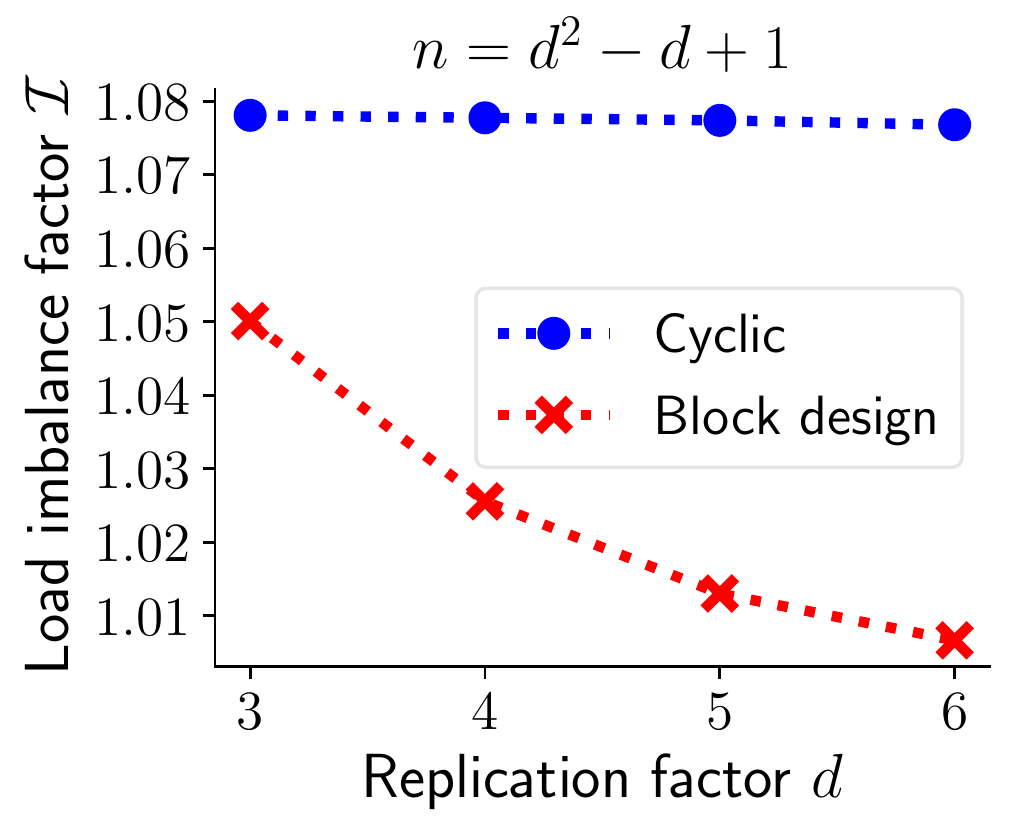}
  \end{subfigure}
  \caption{Simulated average value of $\mathcal{I}$ for $d$-choice allocation with different designs.
  Note that clustering and block designs do not co-exist for the same $d$ and $n$.}
\label{fig:fig_I_cl_vs_cyclic_vs_bibd}
\end{figure}

\subsection{On the Impact of Overlaps Between the Service Choices}
\label{subsec:why_smaller_overlaps_is_better}
Storage redundancy allows the system to split the demand for the popular objects across multiple nodes, hence enabling the system to achieve better load balance across the nodes in the presence of skews in object popularities.
In order to minimize the risk of overburdening a storage node, a natural strategy would be to decouple the overlaps between the service choices ($C_i$) for the objects that are expected to be more popular than others.
In this paper we assume no a priori knowledge on the object popularities; in particular we assume cumulative demand remains constant at $\Sigma$ while all possible object popularity vectors are equally likely, which implies that the object demands are distributed as the uniform spacings within $[0, \Sigma]$ (Sec.~\ref{subsec:sys_model}).
Our model then seeks to answer how one should design the overlaps between the service choices of objects when no a priori knowledge is available on the object popularities.

Recall from \eqref{eq:S_Sigmaum_of_intersection_cardinalities} that all $d$-choice allocations yield the same cumulative overlap between object choices $C_i$ and each design gives a different way of distributing the overlaps across $C_i$.
With simulations (as presented in Fig.~\ref{fig:fig_I_cl_vs_cyclic_vs_bibd}) we found that in order to achieve higher load balancing performance, it is better to spread the overlaps evenly across all pairs of objects (using block design) than distributing them in an unbalanced manner by  implementing smaller overlaps between some objects while implementing larger number of overlaps between others (such as using clustering or cyclic design).
Reducing the overlaps between the service choices for a given set of objects $S$ enlarges the node expansion of $S$ (as explained in detail in Sec.~\ref{subsec:dchoice_alloc_eval}), hence increasing the capacity available for jointly serving the objects within $S$.
However this leads to a reduction in the node expansion for other sets of objects, hence reducing the capacity available for the joint use of those objects.

Overall reducing the service choice overlaps between the objects that are known to be more popular than others will allow the system to balance the offered load, which is expected to be skewed towards the popular objects, more effectively.
However reducing the overlaps for a particular set of objects is risky when we don't know which objects are going to be more popular, because this would increase the overlaps for other sets of objects, one of which might end up being the true set that is more popular than others.
This is exactly the case implemented in our offered load model; few objects will be highly popular while most of them will have average popularity (as implied by \ref{N_a_b_large} and \ref{N_a_b_medium}) and we don't know a priori which of those that are highly popular.
When no information is available on which objects will be more popular, it is not possible to select the popular objects and reduce the overlaps between their service choices.
Then the natural strategy would be to avoid the risk of large overlaps.
Indeed the simulations show for our case with no a priori knowledge on object popularities that allocating the service choices for a group of sets of objects with smaller service choice overlap (as in clustering or cyclic design) performs on average worse than treating all objects the same and minimizing the overlaps across the service choices of all pairs of objects (as in block designs).


The rationale of favoring many but consistently small overlaps over fewer but occasionally larger overlaps has very recently been observed to perform well also in the context of scheduling compute jobs with bi-modal job size distribution.
The authors in \cite{RedundancySchedulingForBiModalJobSizes:BehrouziS19} consider replicating every arriving job (ball) across $r$ nodes (bins), in which the overlaps between the sets of nodes assigned to subsequent jobs impact queueing times at the nodes.
The authors observed that the most effective way to control the overlaps across the subsequent node-assignment rounds is to use a block design, which balances the large jobs across the nodes more effectively than cyclic or random job-to-node assignment strategies.

\begin{figure*}[t]
  \centering
  \begin{subfigure}
    \centering
    \includegraphics[width=.32\textwidth, keepaspectratio=true]{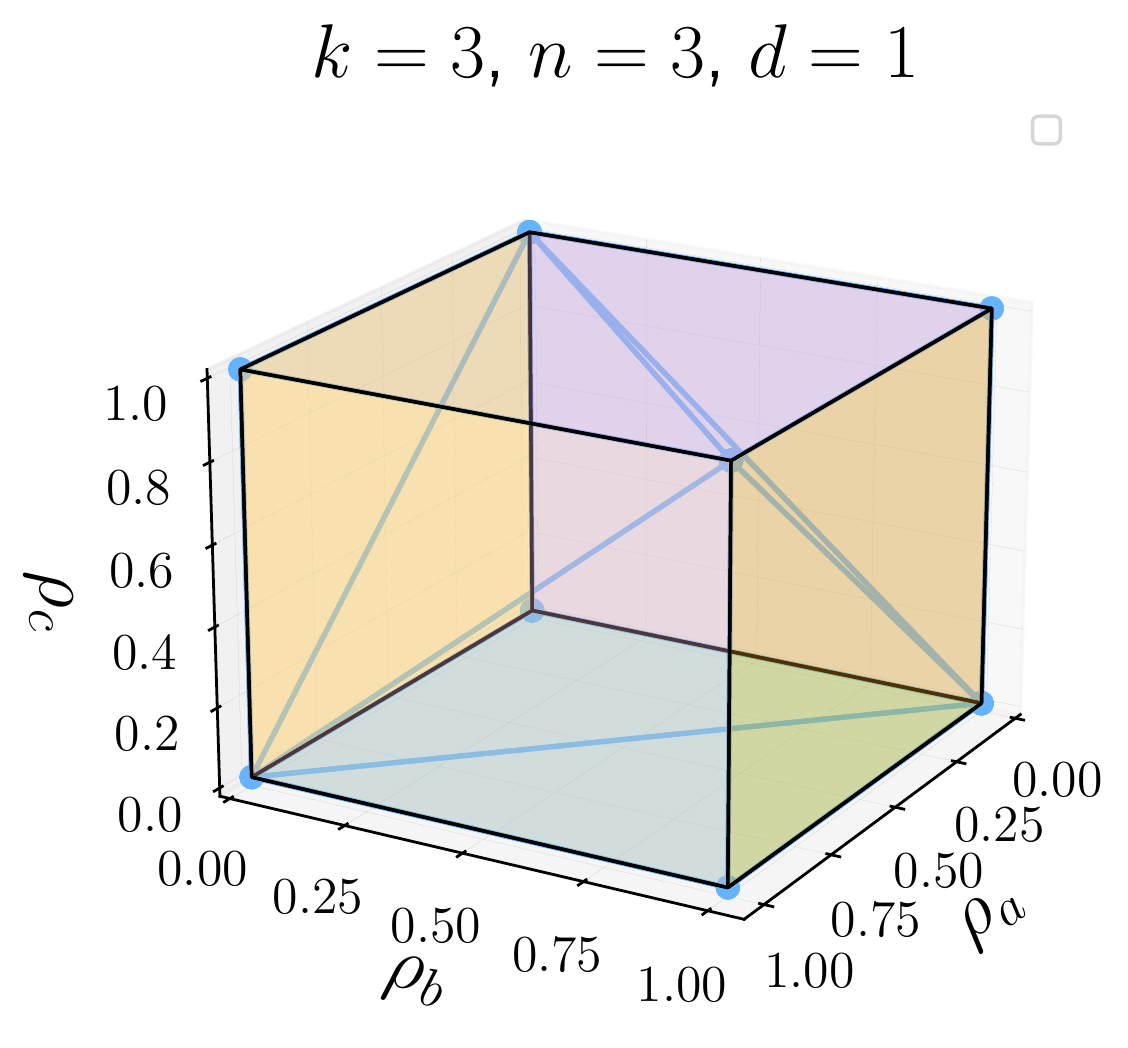}
  \end{subfigure}
  \begin{subfigure}
    \centering
    \includegraphics[width=.32\textwidth, keepaspectratio=true]{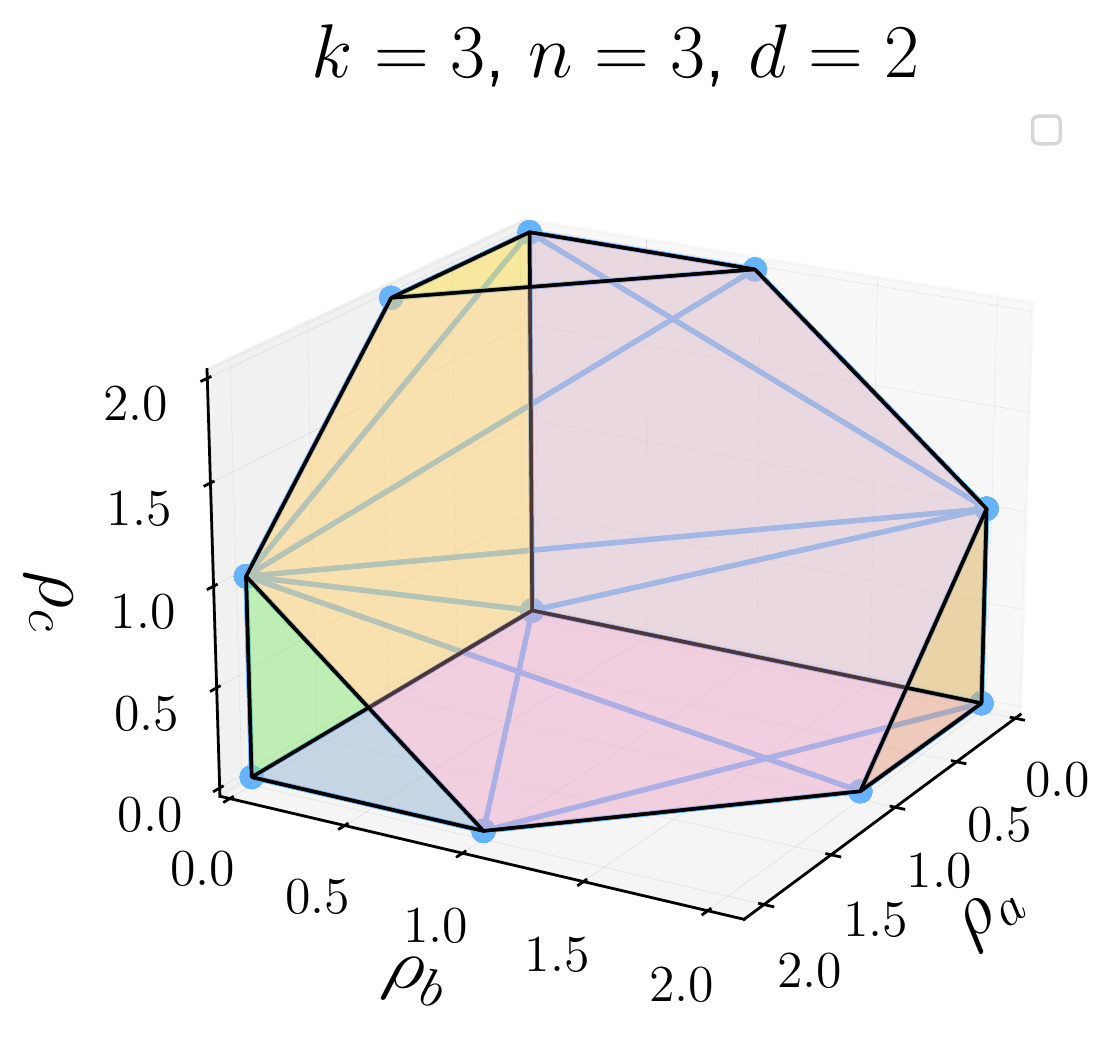}
  \end{subfigure}
  \begin{subfigure}
    \centering
    \includegraphics[width=.32\textwidth, keepaspectratio=true]{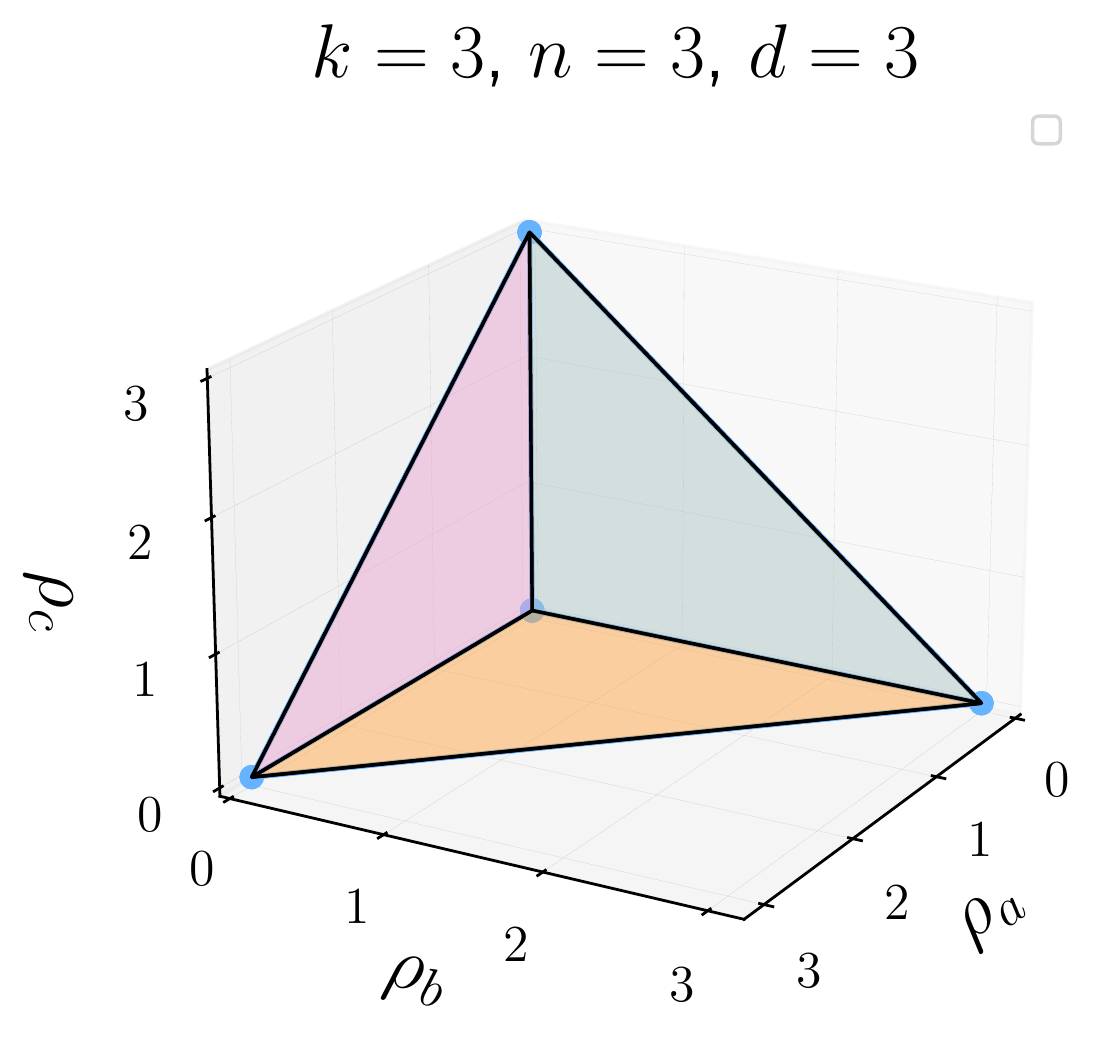}
  \end{subfigure}
  \caption{Service capacity region of regular balanced $d$-choice allocation for $d = 1, 2, 3$.}
  \label{fig:C_dchoice}
\end{figure*}

\section{Interpreting Load Balancing Performance with the Shape of Service Capacity Region}
\label{sec:loadbalancing_w_geometry}
Fig.~\ref{fig:C_dchoice} plots the capacity region $\mathcal{C}$ for a system of three servers and three objects with $d$-choice allocation constructed with cyclic design for $d = 1, 2, 3$ (the cyclic design was introduced in Sec.~\ref{subsec:dchoice_alloc_eval}).
When $d=1$, $\mathcal{C}$ is given by the standard unit cube. Setting $d=2$ extends $\mathcal{C}$ by a unit length at the \emph{skew} corners that lie on coordinate axes. We call them skew corners because the object demand vectors that are close to the corners represent the load scenarios with skewed object popularities.
Setting $d=3$ extends the skew corners by an additional unit of length and yields a simplex capacity region. This implies that the total capacity that is available in the system (which is $3$ in this example) can be arbitrarily used for serving any stored object when $d=3$, i.e., when each object is available at every server.


Previously, we observed that incrementing $d$ from one to two yields the greatest increase in the system's load balancing performance and further increments yield diminishing gains (cf.\ Fig.~\ref{fig:I_wd}).
We now investigate this through the geometric interpretation of $\mathcal{P}_{\Sigma}$ that was given in \eqref{eq:P_Sigma}.
As a Corollary of Lemma~\ref{lm:addingchoice_expands_C}, the capacity region for $d$-choice allocation is contained by that for $(d+1)$-choice allocation. This can be seen for the example given in Fig.~\ref{fig:C_dchoice}.
Recall that $\mathcal{S}_{\Sigma}$ is the $k-1$ dimensional standard simplex of side length $\Sigma$ as defined in \eqref{eq:S_Sigma} and $\mathcal{C}$ is the $k$ dimensional polytope representing the system capacity region.
$\mathcal{P}_{\Sigma}$ is proportional to $\mathrm{Vol}(\mathcal{A})$ where $\mathcal{A} := \mathcal{S}_{\Sigma} \cap \mathcal{C}$ (by \eqref{eq:P_Sigma}), which increases with $d$, hence $\mathcal{P}_{\Sigma}$ increases with $d$.
$\mathcal{A}$ is a $k-1$ dimensional polytope such that $\mathcal{S}_{\Sigma}$ and $\mathcal{A}$ share the same (Chebyshev\footnote{Chebyshev center $c$ of a set $S$ is computed by solving $\min_{c, r}\{r \mid \norm{x - c} \leq r, ~\forall x \in S\}$ }) center. Again examples in Fig.~\ref{fig:C_dchoice} help seeing this.

In order to better understand the effect of incrementing $d$ on the load balancing performance, let us go through the examples given in Fig.~\ref{fig:C_dchoice}.
Suppose $\Sigma = 3$. Then we have $\mathrm{Vol}(\mathcal{S}_{\Sigma}) = 9\sqrt{3}/2$.
When $d=1$, $\mathcal{A} = \{(1,1,1)\}$ and $\mathrm{Vol}(\mathcal{A}) = 0$, hence $\mathcal{P}_{\Sigma} = 0$.
When $d=2$, $\mathcal{A}$ is a polygon with the set of vertices
\[ \Set{(0,1,2), ~(0,2,1), ~(1,0,2), ~(1,2,0), ~(2,0,1), ~(2,1,0)} \]
and $\mathrm{Vol}(\mathcal{A}) = 3\sqrt{3}$, hence $\mathcal{P}_{\Sigma} = 2/3$.
When $d=3$, $\mathcal{A} = \mathcal{S}_{\Sigma}$, hence $\mathcal{P}_{\Sigma} = 1$.
This geometric view can be extended to larger dimensions. Incrementing $d$ extends $\mathcal{C}$ by a unit length in the skew corners, which also expands $\mathcal{A}$. This expansion in $\mathcal{A}$ happens outward from its center equally in every direction when $d$ is small. However, the boundary of $\mathcal{S}_{\Sigma}$ does not allow expansion in every direction beyond a value of $d$. Furthermore, the shape of $\mathcal{S}_{\Sigma}$ causes the expansion per increment in $d$ to diminish in volume as $d$ gets larger.
Thus, as $\mathrm{Vol}(\mathcal{A})$ increases, the increase in $\mathcal{P}_{\Sigma}$ per increment in $d$ diminishes as $d$ gets larger.

\section{\texorpdfstring{$d$}{TEXT}-fold Redundancy with XOR's}
\label{sec:dchoice_wxors}
In this section, we will answer \ref{q:3} posed in the Introduction.
Thus far, we have only considered $d$-choice storage allocations with object replicas.
A replicated copy adds a new service choice for only a single object, while a \emph{coded} copy can add a new choice simultaneously for multiple objects.
When an XOR of $r$ objects (i.e., $r$-XOR) is stored on a node that did not previously host any of the XOR'ed objects, each of the $r$ objects will gain a recovery set, i.e., a set of $r$ nodes that can jointly serve the object of interest.

We here consider the $d$-choice storage allocation with $r$-XOR's, which is implemented by distributing the $k$ exact object copies and $k(d-1)/r$ of their $r$-XOR'ed copies evenly across the storage nodes while complying with Def.~\ref{def:reg_balanced_dchoice_alloc}.
This makes sure that each object can be directly accessed through its exact copy and through $d-1$ recovery sets.
Note that we consider recovery sets that contain a single XOR'ed object, which potentially is less storage efficient than schemes that have been previously proposed based on batch codes \cite{BatchCodesAndTheirApps:IshaiKO04}.
For instance, the 3-choice allocation given in \eqref{eq:3choice_cyclic} with object replicas is implemented with 2-XOR's as
\begin{equation}
  \icol{a\\d+c} \icol{b\\f+g} \icol{c\\a+b} \icol{d\\b+e} \icol{e\\a+d} \icol{f\\g+c} \icol{g\\e+f}.
\label{eq:3choice_w_2xors}
\end{equation}

Allocation with $r$-XOR's reduces the storage overhead multiplicatively by $r$.
However, object access from a recovery set requires downloading an object copy from each of the $r$ nodes that jointly implement the choice, hence download overhead of object recovery grows multiplicatively with $r$.
As a direct consequence of this, the load imbalance factor grows additively with $r$ as stated in the following.
\begin{theorem}
  Consider a system with $d$-choice storage allocation that is created with $r$-XOR's, where $r \geq 2$ is an integer.
  
  When $d = o\left(\log(n)\right)$, the following inequality holds in the limit $n \to \infty$
  \begin{equation}
    \frac{1}{2} \leq \frac{\mathcal{I} \cdot d}{\log(n) + \beta_{n, d}} \leq 1 \quad \text{a.s.},
   \label{eq:I_dchoice_wxors_convergence_as_smalld}
  \end{equation}
  where $\beta_{n, d} = r(d-1)\left(1 + \log_{(2)}(n) - \log\left(1 + r(d-1)\right)\right)$,
  and if $\Sigma_n = b_n \cdot n/\log(n)$ for some sequence $b_n > 0$, then as $n \to \infty$
  \begin{equation}
    \mathcal{P}_{\Sigma_n} \to \begin{cases}
      1 & \limsup b_n/d < 1, \\
      0 & \liminf b_n/2d > 1.
    \end{cases}
  \label{eq:P_dchoice_wxors_convergence_smalld}
  \end{equation}
  
  When $d = c\log(n)$ for some constant $c > 0$, the following inequality holds in the limit $n \to \infty$
  \begin{equation}
    \frac{1}{2} \leq \frac{\mathcal{I}}{(\alpha + 1)\left(\frac{3}{2c\alpha} \cdot \frac{\log_{(2)}(n)}{\log(n)} + r\right)} \leq 1 \quad \text{a.s.},
  \label{eq:I_dchoice_wxors_convergence_as_largerd}
  \end{equation}
  where $\alpha$ is the unique positive solution of $e^{-1/c} = (1 + \alpha)e^{-\alpha}$,
  and if $\Sigma_n = b_n \cdot n/\log(n)$ for some sequence $b_n > 0$, then as $n \to \infty$
  \begin{equation}
    \mathcal{P}_{\Sigma_n} \to \begin{cases}
      1 & \limsup 1.5\tau \cdot b_n/d < 1, \\
      0 & \liminf 0.25\tau \cdot b_n/d > 1,
    \end{cases}
  \label{eq:P_dchoice_wxors_convergence_largerd}
  \end{equation}
  where $\tau = c(1 + \alpha)^2/\alpha$.
\label{thm:P_I_dchoice_wxors}
\end{theorem}
\begin{proof}
  See Appendix~\ref{subsec:proof_thm_P_I_dchoice_wxors}.
\end{proof}

\begin{remark}
  Theorem~\ref{thm:P_I_dchoice_wxors} implies that $d$-choice allocation with $r$-XOR's achieves the same scaling of the load imbalance factor $\mathcal{I}$ in $d$ as if the service choices were created with replicas (as stated in Remark~\ref{rm:dchoice}), while also reducing the storage requirement multiplicatively by $r$.
  However, accessing an object from a recovery set requires downloading $r$ object copies to recover one, thus, increasing the object access overhead multiplicatively by $r$. As a consequence of this, $\mathcal{I}$ in this case increases additively in $r$, which can be seen from its limiting value range given in \eqref{eq:I_dchoice_wxors_convergence_as_smalld} and \eqref{eq:I_dchoice_wxors_convergence_as_largerd}.
\label{rm:dchoice_wxors}
\end{remark}

\subsection{Note on Constructing \texorpdfstring{$d$}{TEXT}-choice Allocations with \texorpdfstring{$r$}{TEXT}-XOR's}
A $d$-choice storage allocation with $r$-XOR's consists of $k$ exact and $k(d-1)/r$ of $r$-XOR'ed object copies, and distributes them across the nodes in a way that complies with the balanced and regular allocation requirements given in Def.~\ref{def:reg_balanced_dchoice_alloc}.
This means each object has $d-1$ XOR'ed choices, thus each object should be a part of $d-1$ different XOR'ed copies.
In addition, sets of objects that are XOR'ed together should not intersect pairwise at more than one object since this would violate the requirement that service choices must be disjoint for each object.

Clearly, $d$-choice allocation with $r$-XOR's does not exist for all values of $k$, $n$, and $d$.
First of all, as described previously in this Section, $k(d-1)/r$ of $r$-XOR'ed object copies are required, which means we need to have $r | k(d-1)$.
Second, the requirement that XOR'ed sets should intersect pairwise at most at one object is similar to a block design. Indeed the 3-choice allocation with 2-XOR's given in \eqref{eq:3choice_w_2xors} is constructed based on a symmetric BIBD with $\lambda=1$ (see Def.~\ref{def:block_design} and the following paragraph).
We do not address the construction of $d$-choice allocations with $r$-XOR's, but only study their load balancing performance by assuming their existence.

\section{Conclusions}
Storage systems need to have the ability to balance the offered load across the storage nodes in order to provide fast and predictable content access performance.
Data objects are replicated across multiple nodes in distributed storage systems to implement robust load balancing in the presence of skews and changes in object popularities.
In this paper, we developed a quantitative answer for two natural questions on implementing resource efficient distributed storage with robust load balancing ability:
1) How does the ability of load balancing improve per added level of storage redundancy for each data object?
2) Can storage efficient alternatives be used instead of replication to improve load balancing?

As an answer for the first question, we found that system's load balancing performance initially improves multiplicatively with the level of added storage redundancy $d$. Somewhat interestingly, once $d$ reaches within a linear range of $\log(\text{total \# of storage nodes})$, system's load balancing performance improves exponentially.
As an answer for the second question, we found that implementing storage redundancy with XOR's of $r$ objects rather than object replicas yield the same improvement in load balancing performance, while also reducing the storage overhead multiplicatively by $r$.
However, accessing data storage by decoding from XOR'ed content requires jointly accessing $r$ storage nodes (in contrast to a replica being available at a single node), which reduces the load balancing performance additively by $r$.

\section{Appendix}
\label{sec:appendix}

The following expression for $M_{k, d=1}$ (maximal spacing within $k$ uniform spacings in the unit line) is well known
\[ M_{k, d=1} = \frac{\max\left\{E_1, \dots, E_k\right\}}{E_1 + \dots + E_k} \text{ in distribution,} \]
where $E_i$ denote i.i.d. unit-mean Exponential random variables (RV's).

Joint distribution of the $k$ uniform spacings on the unit line $(S_1, \dots, S_k)$ is known to be the same as the joint distribution of $(E_1/\Sigma, \dots, E_k/\Sigma)$ where $\Sigma = E_1 + \dots + E_k$.
Using this representation, maximal non-overlapping $m$-spacing within $k$ uniform spacings on the unit line can be expressed as follows.
\begin{lemma}
  For $k = mn$, we have
  \[ M^{(\text{no})}_{k, m} = \frac{\max\left\{\Gamma_1, \dots, \Gamma_n \right\}}{\Gamma_1 + \dots + \Gamma_n} \;\text{in distribution}, \]
  where $\Gamma_i$ are i.i.d. as $\mathrm{Gamma}$ with a shape parameter of $m$ and a rate of $1$, i.e., $\Gamma_i = \sum_{j=1}^m E_j$.
\label{lm:Mnonoverlappingd_def_w_Gammas}
\end{lemma}

\subsection{Proof of Lemma~\ref{lm:dchoice_alloc_is_batchcode}}
\label{subsec:proof_lm_dchoice_alloc_is_batchcode}
\begin{proof}
  As discussed in Sec.~\ref{subsec:dchoice_alloc_eval}, a regular balanced $d$-choice allocation with object replicas defines a balanced $d$ regular bipartite mapping from the set of objects and the set of nodes, which we refer to it as its allocation graph.
  First, by K\H onig's theorem, every regular bipartite graph has a perfect matching, hence the allocation graph has a perfect matching.
  
  Let $S$ be a set of objects and $N(S)$ denote its neighborhood, i.e., the set of all nodes that host at least one of the objects in $S$.
  Since the allocation graph has a perfect matching, by Hall's theorem, we have $|N(S)| \geq |S|$ for every $S$.
  This shows that the storage allocation defines a $(k, kd, n, n, 1)$ batch code.
  Given that the graph is $d$ regular, it can only qualify for a $(k, kd, d, n, 1)$ multiset batch code.
\end{proof}

\subsection{Proof of Corollary~\ref{cor:P_Sigmaless_geq_P_Sigma}}
\label{subsec:proof_cor_P_Sigmaless_geq_P_Sigma}
\begin{proof}
  Recall that the system stores $k$ objects. 
  Let us denote its capacity region with $\mathcal{C}$ and denote its intersection with $\mathcal{S}_{\Sigma}$ as $\mathcal{T}_{\Sigma}$.
  Notice that $\mathcal{S}_{\Sigma^\prime}$ is obtained by scaling $\mathcal{S}_{\Sigma}$ down with $\Sigma^\prime/\Sigma$, hence we have
  \[ \mathrm{Volume}(\mathcal{S}_{\Sigma^\prime})/\mathrm{Volume}(\mathcal{S}_{\Sigma}) = (\Sigma^\prime/\Sigma)^k. \]
  
  For any $\bm{x} \in \mathcal{T}_{\Sigma} \subset \mathcal{C}$, its scaled version $\Sigma^\prime/\Sigma \cdot \bm{x}$ also lies in $\mathcal{C}$. This comes from the convexity of $\mathcal{C}$. (Note that the origin $\bm{0} \in \mathcal{C}$.)
  Let us then scale down $\mathcal{T}_{\Sigma}$ with $\Sigma^\prime/\Sigma$, and denote it with $\mathcal{T}_{\Sigma}^\prime$.
  Then $\mathcal{T}_{\Sigma}^\prime$ will also lie in $\mathcal{C}$.
  We also know that
  \[ \mathrm{Volume}(\mathcal{T}_{\Sigma}^\prime)/\mathrm{Volume}(\mathcal{T}_{\Sigma}) = (\Sigma^\prime/\Sigma)^k. \]
  We have $\mathcal{T}_{\Sigma}^\prime \subseteq \mathcal{C} \bigcap \mathcal{S}_{\Sigma^\prime}$, then
  \begin{equation*}
      \begin{split}
         \mathcal{P}_{\Sigma^\prime} = \frac{\mathrm{Volume}(\mathcal{C} \bigcap \mathcal{S}_{\Sigma^\prime})}{\mathrm{Volume}(\mathcal{S}_{\Sigma^\prime})}
         &\geq \frac{\mathrm{Volume}(\mathcal{T}_{\Sigma}^\prime)}{\mathrm{Volume}(\mathcal{S}_{\Sigma^\prime})}\\
         & = \frac{\mathrm{Volume}(\mathcal{T}_{\Sigma}) \cdot (\Sigma^\prime/\Sigma)^k}{\mathrm{Volume}(\mathcal{S}_{\Sigma}) \cdot (\Sigma^\prime/\Sigma)^k} = \mathcal{P}_{\Sigma}.
      \end{split}
  \end{equation*}
\end{proof}

\subsection{Proof of Lemma~\ref{lm:addingchoice_expands_C}}
\label{subsec:proof_lm_addingchoice_expands_C}
\begin{proof}
  Let the given storage allocation, that yields the capacity region $\mathcal{C}$, be described by the matrices $\bm{T}$ and $\bm{M}$ (in the sense of \eqref{eq:Mx_leq_C}).
  The described modification on the allocation says that a new service choice is added for one of the stored objects by either creating an additional service choice for the object, via replicating it on a node that did not previously host the object or by adding a new choice simultaneously for multiple objects via encoding the objects together, and storing the coded copy on a node that did not previously host any of the encoded objects.
  We consider these two cases separately in the following.
  
  \vspace{1ex}
  \noindent
  \textbf{When the new service choice is created with replication}:
  Let a tagged object be copied to a node that did not previously host the tagged object.
  The newly added choice can be captured by adding a new column to both allocation matrices $\bm{T}$ and $\bm{M}$.
  Without loss of generality, suppose this new column is appended to both matrices at the end.
  Let us denote the modified versions of these matrices as $\bm{T}^\prime$ and $\bm{M}^\prime$ respectively.
  
  First, we show that any point in $\mathcal{C}$ also lies in $\mathcal{C}^\prime$.
  Let us define $D = \left\{\bm{x}~|~\bm{M} \cdot \bm{x} \preceq \bm{1}, ~\bm{x} \succeq \bm{0} \right\}$, and $D^\prime$ similarly with $\bm{M}^\prime$.
  Let $\bm{p} \in \mathcal{C}$, then there is an $\bm{x}$ in $D$ such that $\bm{p} = \bm{T} \cdot \bm{x}$.
  Let us generate $\bm{x}^\prime$ by appending a $0$ at the end of $\bm{x}$. Then, $\bm{x}^\prime \in D^\prime$ since $\bm{M}^\prime \cdot \bm{x}^\prime = \bm{M} \cdot \bm{x} \preceq \bm{1}$, and $\bm{T}^\prime \cdot \bm{x}^\prime = \bm{p}$. Thus, $\bm{p}$ also lies in $\mathcal{C}^\prime$, which implies $\mathcal{C} \subseteq \mathcal{C}^\prime$.
  
  Next, we show that there is at least one point that lies in $\mathcal{C}^\prime$ but not in $\mathcal{C}$. 
  Suppose that the tagged object is stored in $d+1$ nodes after its number of choice is incremented (the modification). Then, the system can supply $(d+1)C$ units of demand for the tagged object and zero demand for all other objects, while it could not supply this before the modification was implemented on the storage allocation. This together with the fact that $\mathcal{C} \subseteq \mathcal{C}^\prime$ implies $\mathcal{C} \subset \mathcal{C}^\prime$.
  
  \vspace{1ex}
  \noindent
  \textbf{When the new service choice is created with coding}:
  Let a new coded object copy be stored on a node that did not previously host any of the objects that constitute the coded copy.
  This adds a new choice for multiple objects simultaneously, which can be captured (as in the case above with replication) by adding new columns to allocation matrices $\bm{T}$ and $\bm{M}$.
  The same arguments used above for the case with replication can be easily repeated here showing $\mathcal{C} \subset \mathcal{C}^\prime$ holds for this case as well.
\end{proof}

\subsection{Proof of Lemma~\ref{lm:Mnonoverlappingd_convergence_as}}
\label{subsec:proof_lm_Mnonoverlappingd_convergence_as}
\begin{proof}
  \noindent
  \textbf{Proof of \eqref{eq:Mnonoverlappingd_convergence_indist}}:
  By Lemma~\ref{lm:Mnonoverlappingd_def_w_Gammas}, we have
  \begin{equation}
    M^{(\text{no})}_{k, m} = \frac{\max\left\{\Gamma_1, \dots, \Gamma_n \right\}}{\Gamma_1 + \dots + \Gamma_n} \;\text{in distribution},
  \label{eq:Mnonoverlappingd_representation_w_Gs}
  \end{equation}
  where $k = m \cdot n$ and $\Gamma_i$'s are i.i.d. as $\mathrm{Gamma}$ with a shape parameter of $m$ and a rate of $1$.
  
  From Darling \cite[Sec. 3]{AsymptoticsForMaxOverSum:Darling52}, we know for fixed $m$ as $n \to \infty$
  \begin{equation}
  \begin{split}
     \Pr&\Biggl\{1/M^{(\text{no})}_{k, m} > m\frac{n}{\log(n)} - (m-1)m\frac{n\log_2(n)}{\log^2(n)} \\
     &\quad + m\log(\Gamma(m))\frac{n}{\log^2(n)} + m\frac{n}{\log^2(n)}x\Biggr\} \to G(x).
   \end{split}
  \label{eq:Darling_1overMnonoverlappingd_convergence_indist}
  \end{equation}
  From this we get
  \begin{equation*}
   \begin{split}
     \Pr&\Biggl\{M^{(\text{no})}_{k, m} < \frac{\log(n)}{mn}\Bigl(1 - (m-1)\frac{\log_2(n)}{\log(n)} \\
     &\quad + \frac{\log(\Gamma(m))}{\log(n)} - \frac{x}{\log(n)} \Bigr)^{-1} \Biggr\} \to G(x).
   \end{split}
  \end{equation*}
  Defining 
  \[ \alpha_n = \frac{x}{\log(n)} + (m-1)\frac{\log_2(n)}{\log(n)} - \frac{\log(\Gamma(m))}{\log(n)}, \]
  we can write
  \[ \Pr\left\{M^{(\text{no})}_{k, m} < \frac{\log(n)}{mn}(1 - \alpha_n)^{-1} \right\} \to G(x). \]
  Using the Taylor expansion on $1/(1 - \alpha_n)$, we can write
  \begin{equation*}
   \begin{split}
     \Pr\Bigl\{M^{(\text{no})}_{k, m} &< \frac{1}{mn}\Bigl(x + \log(n) \\
     &\quad + (m-1)\log_2(n) - \log(\Gamma(m)) \Bigr) \Bigr\} \to G(x),
   \end{split}
  \end{equation*}
  which gives us \eqref{eq:Mnonoverlappingd_convergence_indist}.
  
  \noindent
  \textbf{Proof of \eqref{eq:Mnonoverlappingd_convergence_as}}:

  For the maximal spacing $M_{k, d=1}$ in $k$ uniform spacings on the unit line, results in \cite[Theorem 2.1]{EntropyAndMaximalSpacingsForRandomPartitions:Slud78} show that
  \begin{equation}
    \lim_{k \to \infty} M_{k, d=1} \cdot k/\log(k) = 1 \text{ a.s.}
  \label{eq:M_convergence_as}
  \end{equation}
  The same theorem actually shows that the error in the above convergence is $O(\log_2(k)/\log(k))$ a.s. as $k \to \infty$.
  The presented proof is established from the following
  \begin{equation}
    \Prob{|M_{k, d=1} \cdot k/\log(k) - 1| > \delta_k} = O(k^{-\delta_k})
  \label{eq:M_tailbound}
  \end{equation}
  for any sequence such that $\delta_k \log(k) \to \infty$ and $\delta_k \to 0$ as $k \to \infty$.
  
  Recall that $M^{(\text{no})}_{k, m}$ refers to the maximal non-overlapping $m$-spacing in $k = n \cdot m$ uniform spacings on the unit line.
  By applying the argument that is used to prove \cite[Theorem 2.1]{EntropyAndMaximalSpacingsForRandomPartitions:Slud78}, we here show that a modified version of \eqref{eq:M_convergence_as} holds also for $M^{(\text{no})}_{k, m}$. In the statement of the Lemma, this is expressed in \eqref{eq:Mnonoverlappingd_convergence_as}.
  
  The following bound, which is similar to that given in \cite[Lemma 7]{StrongLawsForMaximalKSpacing:DeheuvelsD84}, will allow us to obtain a result similar to \eqref{eq:M_tailbound}.
  Let $u_k$ be a fixed sequence to be defined later. Using the representation of $M^{(\text{no})}_{k, m}$ that is given in \eqref{eq:Mnonoverlappingd_representation_w_Gs}
  \begin{equation}
  \begin{split}
    \Prob{M^{(\text{no})}_{k, m} &> u_k} \\
    &= \Prob{M^{(\text{no})}_{k, m} > u_k; ~\sum_{i=1}^{n} \Gamma_i \leq k - k^{3/4}} \\
    &\quad + \Prob{M^{(\text{no})}_{k, m} > u_k; ~\sum_{i=1}^{n} \Gamma_i > k - k^{3/4}} \\
    &\leq \Prob{\sum_{i=1}^{n} \Gamma_i \leq k - k^{3/4}} \\
    &\quad + \Prob{\max_{1\le i\le n} \Gamma_i > u_k\left(k - k^{3/4}\right) } \\
    &\stackrel{(a)}{\leq} e^{-\sqrt{k}/2} + n \cdot \gamma_m\bigl[u_k\bigl(k - k^{3/4}\bigr)\bigr]
  \end{split}
  \label{eq:Mnonoverlappingd_tail_upper_bound}
  \end{equation}
  where
  i) $\Gamma_i$'s are i.i.d. as $\mathrm{Gamma}$ with a shape parameter of $m$ and a rate of $1$,
  ii) $(a)$ follows by a large deviation argument on the left side of the expression and a union bound on the right side,
  iii) $\gamma_m$ denotes the tail distribution of $\Gamma_i$ as
  \[ \gamma_m(x) \sim \int_{x}^{\infty} \frac{u^{m-1}}{\Gamma(m)}e^{-u} \dx{u}. \]
  
  For some sequence $\delta_k > 0$, let us set
  \[ u_k = \frac{(1 + \delta_k)\log(k) + (m-1)\log_2(k)}{k}. \]
  We know by \cite[Lemma 5]{StrongLawsForMaximalKSpacing:DeheuvelsD84} that as $x \to \infty$
  \[ \gamma_m(x) \sim \frac{x^{m-1}}{\Gamma(m)}e^{-x}. \]
  Now define $\epsilon_k = k^{-1/4}$ and further suppose that $\delta_k = O(\log_2(k)/\log(k))$. Then we get
  \begin{equation*}
  \begin{split}
    \gamma_m\bigl[u_k\bigl(k &- k^{3/4}\bigr)\bigr] \\
    &\sim \frac{(1 - \epsilon_k)^{m-1}}{\Gamma(m)}\times \\
    &\quad \frac{\left((1+\delta_k)\log(k) + (m-1)\log_2(k)\right))^{m-1}}{{ k^{(1+\delta_k)(1-\epsilon_k)} \log(k)^{(m-1)(1-\epsilon_k)} }} \\
    &\leq O\left(k^{-(1 + 3\delta_k/4)}\right).
  \end{split}
  \end{equation*}
  Substituting this estimate into \eqref{eq:Mnonoverlappingd_tail_upper_bound} we obtain
  \begin{equation}
    \Prob{M^{(\text{no})}_{k, m} > u_k} \leq O\left(n^{-\delta_k/2}\right).
  \label{eq:Mnonoverlappingd_tail_upper_bound_asymptotic}
  \end{equation}
  
  Arguing as in the proof of \cite[Theorem 2.1]{EntropyAndMaximalSpacingsForRandomPartitions:Slud78}, let us define the subsequence $k_t$, $t = 1, 2, \dots$ to be $k_t = \floor{e^{\sqrt{2t}}}$ where rounding is to the largest multiple of $m$.
  Further choose the subsequence $\delta_{k_t}$ to be
  \[ \delta_{k_t} = \frac{\log(t)}{\alpha\sqrt{t}} = O\left(\frac{\log_2(k_t)}{\log(k_t)} \right). \]
  Notice that $\delta_{k_t}$ as given above satisfies our previous assumptions:  $\delta_{k_t} \cdot \log(t) \to \infty$ and $\delta_{k_t} \to 0$ as $t \to \infty$.
  
  For $0 < \alpha < 1/\sqrt{2}$ we have
  \[ \sum_{t} k_t^{-\delta_{k_t}/2} = \sum_{t} e^{-\log(t)/(\alpha\sqrt{2}) } < \infty. \]
  Then by the first Borel-Cantelli lemma, it follows that the inequality
  \[ M^{(\text{no})}_{k_t, m} > \frac{(1 + \delta_{k_t})\log(k_t) + (m-1)\log_2(k_t)}{k_t} \]
  occurs finitely often. Thus we have
  \begin{equation}
    M^{(\text{no})}_{k_t, m} \leq \log(k_t)/k_t + O\left(\log_2(k_t)/k_t\right) \quad \text{a.s.}
  \label{eq:Mnk_upper_bound_as}
  \end{equation}
  
  Given that $(k_{t+1} - k_t) \sim k_t/\log(k_t)$, again from Darling \cite[Sec. 3]{AsymptoticsForMaxOverSum:Darling52} we have the following bound for $0 \leq \ell \leq k_{t+1} - k_t$
  \begin{equation}
    \left|\frac{\log(k_t + \ell)}{k_t + \ell} - \frac{\log(k_t)}{k_t}\right| = O\left(\frac{\ell \cdot \log(k_t)}{k_t \cdot k_{t+1}} \right) \leq O(1/k_t)
  \label{eq:log_diff_upper_bound}
  \end{equation}
  
  Clearly $M^{(\text{no})}_{k_t, m} \geq M^{(\text{no})}_{k_t + \ell, m}$.
  Therefore it follows from \eqref{eq:Mnk_upper_bound_as} and \eqref{eq:log_diff_upper_bound} that
  \[ M^{(\text{no})}_{k, m} \leq \log(k)/k + O\left(\log_2(k)/k\right)  \quad \text{a.s.} \]
  This shows \eqref{eq:Mnonoverlappingd_convergence_as_werr}.
  
  We can therefore conclude that
  \begin{equation}
    \limsup_{k \to \infty} M^{(\text{no})}_{k, m} \cdot k/\log(k) \leq 1.
  \label{eq:Mnonoverlappingd_limsup_upper_bound}
  \end{equation}
  
  Given that $M^{(\text{no})}_{k, m} \geq M_{k, d=1}$, from \eqref{eq:M_convergence_as} it immediately follows
  \[ \liminf_{k \to \infty} M^{(\text{no})}_{k, m} \cdot k/\log(k) \geq 1. \]
  
  This together with \eqref{eq:Mnonoverlappingd_limsup_upper_bound} implies that as $k \to \infty$
  \[ M^{(\text{no})}_{k, m} \cdot k/\log(k) \to 1 \quad \text{a.s.} \]
  
  This gives us \eqref{eq:Mnonoverlappingd_convergence_as}.
\end{proof}

\subsection{Maximal \texorpdfstring{$d$}{TEXT}-spacing on the Unit Circle}
\label{subsec:proof_Mcirculard}
We here show that the maximal $d$-spacing $M_{k, d}^{(c)}$ defined for $k$ ordered uniform samples on the unit circle (see Def.~\ref{def:Mcirculard}) converge to its counterpart $M_{k, d}$ defined on the unit line.
In the following, we show convergence first in distribution, then in probability, and finally almost surely.
Note that showing almost sure convergence implies convergence in probability, which then implies convergence in distribution.
Convergence in this order is presented so as to make the arguments transparent.

\begin{lemma}
  For $d < k$,
  \begin{equation*}
    \Prob{M_{k, d} > x} \leq \Prob{M_{k, d}^{(c)} > x} \leq \frac{k}{k-d} \Prob{M_{k, d} > x}.
  \end{equation*}
\label{eq:Pr_Mcirculard_leq_scaled_Pr_Md}
\end{lemma}
\begin{proof}
  Let us denote the events $\left\{M_{k, d} > x\right\}$ and $\left\{M_{k, d}^{(c)} > x\right\}$ respectively with $L$ and $C$.
  
  The first inequality is immediate; if a sequence of spacings $s = (s_1, s_2, \dots, s_k) \in L$ then $s \in C$, while the opposite direction may not hold. Thus, $L \subseteq C$, hence $\Pr\{L\} \leq \Pr\{C\}$.
  
  Next we show the second inequality.
  Let $s \in L$. Then, at least $k-d$ different permutations of $s$ lie in $L$. In order to see this, let the maximal $d$-spacing within $s$ be $\bm{m} = (s_i, \dots, s_{i+d-1})$. Shifting (by feeding what is shifted out back in the sequence at the opposite end) $s$ to the left by at most $i-1$ times will preserve $\bm{m}$, hence each of the $i-1$ shifted versions will also lie in $L$. Similarly, shifting $s$ to the right by at most $k - (i+d-1)$ times will also preserve $\bm{m}$.
  We call such permutations, which are obtained by shifting with wrapping around, a \textit{cyclic permutation}.
  
  Let us introduce a set $L^\prime \subset L$ such that for any $s \in L^\prime$, no cyclic permutation of $s$ lies in $L^\prime$. $L$ contains at least $k-d$ cyclic permutations of every $s \in L^\prime$. This together with the fact that all sequences of spacings are equally likely (Lemma~\ref{lm:spacings_uniformlydisted}) gives us $(k-d) \Pr\{L^\prime\} \leq \Pr\{L\}$.
  
  Now let $s^\prime \in C$. All $k-1$ cyclic permutations of $s^\prime$ will also lie in $C$ (recall that we are now working on the unit circle).
  This together with the fact $L^\prime \subset L \subseteq C$ and Lemma~\ref{lm:spacings_uniformlydisted} gives us $k \cdot \Pr\{L^\prime\} = \Pr\{C\}$. Putting it all together, we have $\Pr\{C\}/k = \Pr\{L^\prime\} \leq \Pr\{L\}/(k-d)$, which yields the second inequality.
  
  A simpler way to find the second inequality in Lemma~\ref{eq:Pr_Mcirculard_leq_scaled_Pr_Md} is given as follows.
  Recall that the uniform samples, together with the $0$ point, are ordered on the unit circle as $0, U_{(1)}, \dots, U_{(k-1)}$.
  Let us denote the index of the sample at which the maximal $d$-spacings starts with $I$, e.g., $I = i$ means that the maximal $d$-spacing starts at the $i$th minimum uniform sample, $I = 0$ means it starts at the point of $0$.
  We have
  \[ \Prob{M_{k,d} > x} \geq \Prob{M_{k,d}^{(c)} > x; I \leq k - d + 1} \]
  since the event on the right implies the event on the left.
  The right hand side of this inequality can be written as
  \[ \Prob{M_{k,d}^{(c)} > x} \cdot (k-d)/k, \]
  using the independence of the events and the fact that $I$ is uniform on $1, \dots, k$.
\end{proof}


\begin{lemma}
  For $d = o(k)$, $M_{k, d}^{(c)}/M_{k, d} \to 1$ in probability as $k \to \infty$.
\label{lm:Mcirculard_convergence_inprob}
\end{lemma}
\begin{proof}
  It is easy to see $M_{k, d}^{(c)} \geq M_{k, d}$.
  Let $D = M_{k, d}^{(c)} - M_{k, d}$ and $S$ be the set of all sequence of spacings for which $D > 0$.
  For every $s \in S$, $d-1$ of its cyclic permutations (see the Proof of Lemma~\ref{eq:Pr_Mcirculard_leq_scaled_Pr_Md} for the definition of a cyclic permutation) also lie in $S$ while the remaining $k-d$ of them lie in $S^c$ (complement of $S$).
  Thus, for every $d$ points in $S$, there are at least $k-d$ points in $S^c$, and all the points in $S$ or $S^c$ (i.e., all spacings) have the same probability measure (by Lemma~\ref{lm:spacings_uniformlydisted}). This gives us the following upper bound $\Pr\{D > 0\} = \Pr\{S\} \leq d/k$, which $\to 0$ as $k \to \infty$. This implies $M_{k, d}^{(c)}/M_{k, d} \to 1$ in probability.
\end{proof}
In order to use the results known for the convergence of $M_{k, d}$ in probability or a.s. in addressing $M_{k, d}^{(c)}$, we need the following Lemma.

  

\begin{lemma}
  For $d = o(k)$, $M_{k, d}^{(c)}/M_{k, d} \to 1$ a.s. as $k \to \infty$.
\label{lm:Mcirculard_to_Md_as}
\end{lemma}

Before we move on with the proof of Lemma~\ref{lm:Mcirculard_to_Md_as}, we next express the maximal $d$-spacing $M_{k, d}^{(c)}$ on the unit circle in terms of the two different instances of its counterpart defined on the unit line.

Let $2m+1 \geq 1$ be arbitrary and place $2m+1$ i.i.d. uniform random variables on the unit circle with $0 \leq U_{(1)} \leq \ldots \leq U_{(2m+1)} \leq 1$ where the points 0 and 1 are identified. Let $O$ denote the linear sequence starting at 0 and $P$ be the linear sequence starting at $U_{(m+1)}$ the median, without loss of generality.
We know $U_{(m+1)} = a \in (0,1)$ almost surely. By adding further i.i.d. uniform variates we get two sequences of uniform spacings with parameter $k \geq 2m + 1$, where the first starts at $O$ and the second at $P$.


Let $d_k = o(k)$ be a sequence of $d_k$-spacings for the $k$th realisation.
Let $M^{(c)}_{k, d_k}$ be the maximal \emph{circular} $d_k$-spacing on the previously constructed unit circle (as defined in Def.~\ref{def:Mcirculard}), and let $M^{(O)}_{k, d_k}$, $M^{(P)}_{k, d_k}$ be the maximal $d_k$-spacing for the line segments that stretch along the sequences $O$ and $P$ respectively.
We say that the circular spacings are {\em covered} by $O$ and $P$ if any circular $d_k$-spacing on the circle is either a $d_k$-spacing for $O$ or for $P$ (or both). This will always be the case if the number of intervening points $N_k$ going from the beginning of $O$ to the beginning of $P$ clockwise is such that $N_k \geq d_k$ and also for the number of points $M_k$ going from the end of $P$ to the end of $O$ clockwise. Clearly if the circle spacings are covered by $O$ and $P$,
\begin{equation}
  M^{(c)}_{k, d_k} \leq \max\left\{M_{k, d_k}^{(O)}, M_{k, d_k}^{(P)}\right\}.
\label{eq:Mcirculard_leq_max_MOd_MPd}
\end{equation}

We now show that a.s. for any sequence there is a $K_d$ sufficiently large so that $\forall k \geq K_d$ it holds that $N_k, M_k \geq d_k$.
It is enough to show this for $N_k$ as the same argument will apply to $M_k$.

The interval from the beginning of $O$ to the beginning of $P$ has length $a$ (recall $U_{(m+1)} = a$) and therefore
\[ N_k/k \rightarrow a \quad\text{a.s.} \]
by the strong law of large numbers. Therefore $\exists K$ such that $\forall k \geq K$, $N_k \geq \frac{a \times k}{2}$ and $a > 0$ a.s. Since $d_k = o(k)$ it follows that $\exists K_N \geq K$ such that $N_k \geq d_k, k \geq K_N$. By the same argument $\exists K_M$ such that $M_k \geq d_k, k \geq K_M$.
Now $K_d = \max\{K_N, K_M\}$  is the required number and it follows that inequality \eqref{eq:Mcirculard_leq_max_MOd_MPd} holds $\forall k \geq K_d$ a.s.

Now we are ready to prove Lemma~\ref{lm:Mcirculard_to_Md_as}, that is to show $M_{k, d}^{(c)}/M_{k, d} \to 1$ a.s.
\begin{proof}
  First, given that $M_{k, d}^{(c)} \geq M_{k, d}$, we have
  \begin{equation}
    \liminf\limits_{k \to \infty} M_{k, d}^{(c)}/M_{k, d} \geq 1.
  \label{eq:liminf_Mcd_over_Md_geq_1}
  \end{equation}
  Next using \eqref{eq:Mcirculard_leq_max_MOd_MPd}, we have
  \[ M_{k, d}^{(c)} \leq \max\left\{M_{k, d_k}^{(O)}, M_{k, d_k}^{(P)}\right\}. \]
  This allows us to find
  \begin{equation*}
      \begin{split}
           \limsup\limits_{k \to \infty} & M_{k, d}^{(c)}/M_{k, d} \\
           &\leq \max\Big\{\limsup\limits_{k \to \infty} M_{k, d}^{(O)}/M_{k, d}, \quad\limsup\limits_{k \to \infty} M_{k, d}^{(P)}/M_{k, d}. \Big\}
      \end{split}
  \end{equation*}
  This, together with the fact that $M_{k, d}^{(O)}$ and $M_{k, d}^{(P)}$ converge to $M_{k, d}$ a.s., gives us
  \begin{equation}
    \limsup\limits_{k \to \infty} M_{k, d}^{(c)}/M_{k, d} \leq 1.
  \label{eq:limsup_Mcd_over_Md_leq_1}
  \end{equation}
  Putting \eqref{eq:liminf_Mcd_over_Md_geq_1} and \eqref{eq:limsup_Mcd_over_Md_leq_1} together completes the proof.
\end{proof}

In the following, we show that the results given in Lemma~\ref{lm:Md_convergence_indist} and Lemma~\ref{lm:Md_convergence_as} for $M_{k, d}$ carry over to $M_{k, d}^{(c)}$.
\begin{lemma}
  For any integer $d \geq 1$, as $k \to \infty$
  \begin{equation}
  \begin{split}
    \Pr\Big\{M_{k, d}^{(c)} \cdot k - \log(k) &- (d-1)\log_2(k) \\
    &+ \log((d-1)!) \leq x\Big\} \to G(x).
  \end{split}
  \label{eq:Mcirculard_convergence_indist}
  \end{equation}
  That is the distribution of $M_{k, d}^{(c)} \cdot k - \log(k) - (d-1)\log_2(k) + \log((d-1)!)$ converges to standard Gumbel distribution as $k \to \infty$.
\end{lemma}
\begin{proof}
  
  
  Let us first denote the event that the maximal $d$-spacing on the unit circle lies between the 1st and $k$th uniform sample with $E$, meaning that $(M_{k, d}^{(c)} \mid E) = M_{k, d}$.
  Since the maximal spacing is equally likely to start at any one of the uniform samples $U_0, \dots, U_{k-1}$, we have $\Prob{E} = (k-d)/k$.
  Let us also define
  \[ f_k = \log(k) + (d-1)\log_2(k) - \log((d-1)!). \]
  By the law of total probability
  \begin{equation}
  \begin{split}
    \Prob{M_{k, d}^{(c)} &\cdot k - \log(k) - f_k < x} \\
    &= \Prob{M_{k, d}^{(c)} \cdot k - \log(k) - f_k < x; ~E} \\
    &\quad + \Prob{M_{k, d}^{(c)} \cdot k - \log(k) - f_k < x; ~E^c}.
  \end{split}
  \label{eq:Mcirculard_scaled_l_x}
  \end{equation}
  Left hand side of the sum above can be bounded as
  \begin{equation*}
  \begin{split}
    \Prob{M_{k, d}^{(c)} &\cdot k - \log(k) - f_k < x; ~E} \\ &\stackrel{(a)}{=} \Prob{M_{k, d} \cdot k - \log(k) - f_k < x; ~E} \\
    &\stackrel{(b)}{\geq} \Prob{M_{k, d} \cdot k - \log(k) - f_k < x} - \Prob{E^c},
  \end{split}
  \end{equation*}
  where $(a)$ follows from $(M_{k, d}^{(c)}; E) = (M_{k, d}; E)$,
  and $(b)$ comes from the inequality for events $A$ and $B$
  \[ \Prob{A; B} = \Prob{B} - \Prob{B; A^c} \geq \Prob{B} - \Prob{A^c}. \]
  Putting this in \eqref{eq:Mcirculard_scaled_l_x} gives us
  \begin{equation*}
  \begin{split}
    \Prob{M_{k, d}^{(c)} \cdot k &- \log(k) - f_k < x} \\
    &\geq \Prob{M_{k, d} \cdot k - \log(k) - f_k < x} \\
    &\quad - \Prob{E^c} \\
    &\quad + \Prob{M_{k, d}^{(c)} \cdot k - \log(k) - f_k < x; ~E^c},
  \end{split}
  \end{equation*}
  since $\Prob{E^c} = d/k \to 0$ as $k \to \infty$.
  Overall this gives us
  \begin{equation*}
  \begin{split}
    \liminf\limits_{k \to \infty} \Prob{M_{k, d}^{(c)} &\cdot k - \log(k) - f_k < x} \\
    & \geq \liminf\limits_{k \to \infty} \Prob{M_{k, d} \cdot k - \log(k) - f_k < x}.
  \end{split}
  \end{equation*}
  Given that $M_{k, d}^{(c)} \geq M_{k, d}$, we have the lower bound
  \begin{equation*}
  \begin{split}
    \limsup\limits_{k \to \infty} \Prob{& M_{k, d}^{(c)} \cdot k - \log(k) - f_k < x} \\
    &\leq \limsup\limits_{k \to \infty} \Prob{M_{k, d} \cdot k - \log(k) - f_k < x}.
  \end{split}
  \end{equation*}
  Both the lower and upper bounds given above are equal to $G(x)$ by Lemma~\ref{lm:Md_convergence_indist}, hence showing \eqref{eq:Mcirculard_convergence_indist}.
  
\end{proof}

\begin{lemma}
  For $d = o(\log(k))$, as $k \to \infty$
  \begin{equation*}
    \frac{M_{k, d}^{(c)} \cdot k - \log(k)}{(d-1)\left(1 + \log_2(k) - \log(d)\right)} \to 1 ~~\text{a.s.}
  \end{equation*}
\label{lm:Mcirculard_convergence_as}
\end{lemma}
\begin{proof}
  For brevity, let us define the function
  \[ f(x) = \frac{x \cdot k - \log(k)}{(d-1)\left(1 + \log_2(k) - \log(d)\right)}. \]

  The fact that $M_{k, d}^{(c)} \geq M_{k, d}$ gives us
  \[ f\left(M_{k, d}^{(c)}\right) \geq f\left(M_{k, d}\right) \]
  for $k$ sufficiently large a.s.
  By \eqref{eq:Md_convergence_as_smalld} in Lemma~\ref{lm:Md_convergence_as}, the right hand side of the above inequality $\to 1$ as $k \to \infty$ a.s. Then we have
  \begin{equation}
    \liminf_{k \to \infty} f\left(M_{k, d}^{(c)}\right) \geq 1. 
  \label{eq:liminf_fMcirculard_geq_1}
  \end{equation}
  
  Inequality \eqref{eq:Mcirculard_leq_max_MOd_MPd} gives us
  \[ f\left(M_{k, d}^{(c)}\right) \leq f\left(\max\left\{M_{k, d_k}^{(O)}, M_{k, d_k}^{(P)}\right\}\right) \]
  for $k$ sufficiently large a.s.
  This implies that
  \begin{equation*}
  \begin{split}
    \limsup_{k \to \infty} f&\left(M_{k, d}^{(c)}\right) \\
    & \leq \max\left\{\limsup_{k \to \infty} f\left(M_{k, d}^{(O)}\right), \quad\limsup_{k \to \infty} f\left(M_{k, d}^{(O)}\right) \right\}.
  \end{split}
  \end{equation*}
  
  By \eqref{eq:Md_convergence_as_smalld} in Lemma~\ref{lm:Md_convergence_as},
  \[ \limsup_{k \to \infty} f\left(M_{k, d}^{(O)}\right) = 1 \quad\text{and}\quad \limsup_{k \to \infty} f\left(M_{k, d}^{(P)}\right) = 1 ~~\text{a.s.}. \]
  Hence we have 
  \[ \limsup_{k \to \infty} f\left(M_{k, d}^{(c)}\right) \leq 1. \]
  This together with \eqref{eq:liminf_fMcirculard_geq_1} concludes the proof.
\end{proof}

\begin{lemma}
  For $d = c\log(k) + o(\log_2(k))$ with some constant $c > 0$,
  \[ V_k = \frac{M_{k, d}^{(c)} \cdot k - (1 + \alpha)c \cdot \log(k)}{\log_2(k)} \]
  satisfies
  \begin{equation}
  \begin{split}
     \limsup\limits_{k \to \infty} V_k &= c^\ast (1 + \alpha)/\alpha ~~\text{a.s.} \\
     \liminf\limits_{k \to \infty} V_k &= -c^\dagger (1 + \alpha)/\alpha ~~\text{a.s.}
  \end{split}
  \label{eq:Mcirculard_convergence_as_largerd}
  \end{equation}
  where $\alpha$ is the unique positive solution of $e^{-1/c} = (1 + \alpha)e^{-\alpha}$, and $c^\ast$ and $c^\dagger$ are constants taking values in $[-0.5, 1.5]$ and $[-1.5, -0.5]$ respectively.
\end{lemma}
\begin{proof}
  Shown applying the same ideas used in the proof of Lemma~\ref{lm:Mcirculard_convergence_as} given above.
\end{proof}

\subsection{Proof of Lemma~\ref{lm:on_rgap}}
\label{subsec:proof_lm_on_rgap}
\begin{proof}
  Recall that nodes are indexed by the index of the primary object copies they store, i.e., node $s_i$ stores the primary copy for object $o_i$ for $i = 1, \dots, k$. We denote the set of choices for $o_i$ with $C_i$.
  
  \vspace{.5em}
  \noindent
  \textbf{Proof of $r \geq d-1$:}
  We first prove that $r \geq d-1$ by contradiction.
  Suppose $r < d-1$. Pick an arbitrary object $o_i$ with the set of choices $C_i$.
  Then $o_{i-r}$ is co-located together with $o_i$ on one of the nodes in $C_i$, which we refer to as $s^*$.
  Given that $s^*$ is a choice for $o_i$, any object stored on it must be in the set $\{o_{i-r}, o_{i-r+1}, \dots, o_{i+r}\}$.
  By now $s^*$ stores $o_i$ and $o_{i-r}$ by design, and we now look at the remaining $d-2$ storage slots of $s^*$.
  Given that $s^* \in C_i$, $r$-gap allocation dictates that $s^*$ can store only objects within $\{o_{i-r}, \dots, o_{i+r}\}$.
  We assumed that $s^*$ stores $o_{i-r}$ which means $s^* \in C_{i-r}$ as well. So if $s^*$ stores any object in $\{o_{i+1}, \dots, o_{i+r}\}$ then $C_{i-r} \cap C_{j} \neq \emptyset$ for some $j > i$, which would violate the definition of $r$-gap design.
  Therefore all the remaining $d-2$ storage slots of $s^*$ must be occupied by the objects in the set $\mathcal{O} = \{o_{i-r+1}, \dots, o_{i-1}\}$, which means there needs to be at least $d-2$ different objects within $\mathcal{O}$ (the same object cannot be stored multiple times on the same node), implying $r \geq d-1$.
  
  Perhaps, an easier way to show this is given as follows.
  The objects which may share a node with an object $i$ are those within a set
  \[ S_i = \{r_{\min} \leq i \leq r_{\max} \} \]
  where $r_{\max} - r_{\min} \leq r$ and we apply arithmetic $\mod n$.
  Such sets always contain at most $r+1$ elements.
  For example if $r=2$ and $i=3$, we may take the set $\{2, 3, 4\}$, if $n \geq 5$, which has $3 = r+1$ elements. Or if $n=6$, $i=5$ and $r=4$ we may take the set $\{5, 6, 1, 2, 3\}$ and note that $3-5+6 = 4$ using arithmetic $\mod 6$. The set contains $5$ elements.
  
  Now consider an object $i$ and the corresponding node set $C_i$. It has $d^2$ slots which have to be occupied. Since all sets containing $i$ have at most $r+1$ elements under an $r$-gap design, it must be the case that $(r+1)d \geq d^2$ which implies $r \geq d-1$.
  
  \vspace{.5em}
  \noindent
  \textbf{Proof of the lower bound for the node expansion of sets of objects:}
  We next show that $x \leq N(S) \leq x + 2r$ for any $S = \{o_i, \dots, o_{i+x-1}\}$ for $i = 1, \dots, n$.
  Storage allocation defines a regular bipartite graph, then by Hall's theorem we have $|N(S)| \geq x$.
  The copies of $o_i$ can expand across at most $n_{i-r}, \dots, n_{i+r}$, and the copies of $o_{i+x-1}$ can expand at most across $n_{i+x-1-r}, \dots, n_{i+x-1+r}$.
  Then $S$ can expand at most across $n_{i-r}, \dots, n_{i+x-1+r}$, meaning $|N(S)| \leq x + 2r$.
\end{proof}

\subsection{Proof of Lemma~\ref{lm:rgap_neccsuffcond_for_stability}}
\label{subsec:proof_lm_rgap_neccsuffcond_for_stability}
\begin{proof}
  \noindent
  \textbf{Necessary condition:}
  System is surely unstable if a set of objects $S$ has a cumulative offered load larger than $|N(S)|$.
  Lemma~\ref{lm:on_rgap} states that every consecutive $i$ objects expands across at most $i + 2r$ nodes, meaning that the system can possibly be made stable only if the cumulative offered load for any $i$ consecutive objects is less than $i+2r$, which is exactly what is expressed in \eqref{eq:rgap_necccond_for_stability}.
  
  \vspace{1ex}
  \noindent
  \textbf{Sufficient condition:}
  Suppose that the maximum offered load on any $r$ consecutive objects is $d$, which can be described with the maximal $r$-spacing as $M^{(c)}_{n, r} \cdot \Sigma \leq d$ (recall $\Sigma$ is the cumulative offered load on the system).
  
  Let $x$ be an integer in $[1, n]$.
  Consider the following spiky load scenario starting at $o_x$; offered load $\rho_i$ for $o_i$ is $d$ when $i = x + (r+1)j, ~j=0, 1, \dots, \floor{n/(r+1)}$ and $0$ otherwise.
  In this case, an offered load of magnitude $d$ for each spiky object $o_i$ can be supplied by using up the capacity in all the nodes available in its set of $d$ choices $C_i$ since all other objects that overlap with $o_i$ in their service choices have $0$ offered load (by the $r$-gap design property).
  The system can supply the spiky load regardless of the value for $x$.
  Given that the system's service capacity region is convex (Lemma~\ref{lm:cap_region_is_convex}), any convex combination of any set of spiky load scenarios can also be supplied by the system.
  This can be expressed as follows: system can operate under stability as long as the offered load on every $r+1$ consecutive objects is at most $d$, which implies \eqref{eq:rgap_suffcond_for_stability}.
\end{proof}

\subsection{Proof of Lemma~\ref{lm:clustering_cyclic_neccsuffcond_for_stability}}
\label{subsec:proof_lm_clustering_cyclic_neccsuffcond_for_stability}
\begin{proof}
  Lower bounds come from substituting $r = d-1$ in those given in Lemma~\ref{lm:rgap_neccsuffcond_for_stability}.
  Upper bounds come from observing that every $d+1$ consecutive objects expand to at most $2d$ nodes in the design with clustering, and every $d$ consecutive objects expand to $2d-1$ nodes in the design with cyclic construction.
\end{proof}

\subsection{Proof of Theorem~\ref{thm:P_I_d1}}
\label{subsec:proof_thm_P_I_d1}
\begin{proof}
  Recall that the load at the maximally loaded node $l^{\max}_n$ is given by $M^{(n)}_{k, m} \cdot \Sigma_n$. 
  Almost sure convergence given in \eqref{eq:Mnonoverlappingd_convergence_as} implies for $\Sigma_n = b_n \cdot n/\log(n)$ that
\[ l^{\max}_n \cdot m/b_n \to 1 ~\text{a.s.} \]
  This implies in the limit $n \to \infty$ for any $\delta > 0$
  \begin{equation*}
  \begin{split}
    \Prob{|l^{\max}_n \cdot m/b_n &- 1| > \delta} \\
    &= \Prob{l^{\max}_n > b_n/m \cdot (1 + \delta) } \\
    &\quad + \Prob{l^{\max}_n < b_n/m \cdot (1 - \delta)} \to 0.
  \end{split}
  \end{equation*}
  Given that both terms in the sum above is non-negative, we have
  \begin{equation*}
      \begin{split}
           &\Prob{l^{\max}_n > b_n/m \cdot (1 + \delta) } \to 0,\\
           &\Prob{l^{\max}_n < b_n/m \cdot (1 - \delta)} \to 0.
      \end{split}
  \end{equation*}
  Recall from \eqref{eq:P_I_d1_exact} that $\mathcal{P}_{\Sigma_n}$ is given by $\Prob{l^{\max}_n < 1}$.
  Then the convergence of probabilities given above implies \eqref{eq:P_d1_convergence}.
  
  \eqref{eq:I_d1_convergence_indist} and \eqref{eq:I_d1_convergence_as} come from
  substituting $\mathcal{I} = M^{(n)}_{k, m} \cdot n$ (by \eqref{eq:P_I_d1_exact}) in the convergence results given in Lemma~\ref{lm:Mnonoverlappingd_convergence_as}.
 \end{proof}

\subsection{Proof of Theorem~\ref{thm:P_I_dchoice_w_clustering_cyclic}}
\label{subsec:proof_thm_P_I_dchoice_w_clustering_cyclic}
\begin{proof}
  We first need to recall Lemma~\ref{lm:clustering_cyclic_neccsuffcond_for_stability}; under a cumulative demand of $\Sigma$, $M^{(c)}_{n, d} \cdot \Sigma \leq d$ is sufficient and $M^{(c)}_{n, d} \cdot \Sigma \leq 2d$ is necessary for system stability.
  Here we will refer to $d$ as $d_n$ to make it explicit that it is a sequence in $n$.
  
  \vspace{0.5em}
  \noindent
  \textbf{Proof of \eqref{eq:P_dchoice_w_clustering_cyclic_convergence_smalld}}:
  In this case $d = o\left(\log(n)\right)$.
  Almost sure convergence given in \eqref{eq:Md_convergence_as_smalld} together with Lemma~\ref{lm:Mcirculard_to_Md_as} implies for $\Sigma_n = b_n \cdot n/\log(n)$ that
  \[ M^{(c)}_{n, d} \cdot \Sigma_n/b_n \to 1 ~\text{a.s.} \]
  Recall that $M^{(c)}_{n, d} \cdot \Sigma_n / d_n \leq 1$ is sufficient and $M^{(c)}_{n, d} \cdot \Sigma_n / 2d_n \leq 1$ is necessary for system stability, which respectively implies that $\mathcal{P}_{\Sigma_n} \to 1$ if $\limsup_{n \to \infty} b_n/d_n < 1$, and $\mathcal{P}_{\Sigma_n} \to 0$ if $\liminf_{n \to \infty} b_n/2d_n > 1$, hence \eqref{eq:P_dchoice_w_clustering_cyclic_convergence_smalld}.
  
  \vspace{0.5em}
  \noindent
  \textbf{Proof of \eqref{eq:P_dchoice_w_clustering_cyclic_convergence_largerd}}:
  In this case $d = o\left(\log(n)\right)$.
  Almost sure convergence given in \eqref{eq:Md_convergence_as_largerd} together with Lemma~\ref{lm:Mcirculard_to_Md_as} implies for $\Sigma_n = b_n \cdot n/\log(n)$ that in the limit $n \to \infty$ we have almost surely
  \[ 0.5\tau \leq M^{(c)}_{n, d} \cdot \Sigma_n/b_n \leq 1.5\tau. \]
  Then the necessary and sufficient conditions (as used in the previous step while showing \eqref{eq:P_dchoice_w_clustering_cyclic_convergence_smalld}) for system stability imply \eqref{eq:P_dchoice_w_clustering_cyclic_convergence_largerd}.
  
  \vspace{0.5em}
  \noindent
  \textbf{Proof of \eqref{eq:I_dchoice_w_clustering_or_cyclic_convergence_as_smalld} and \eqref{eq:I_dchoice_w_clustering_or_cyclic_convergence_as_largerd}}:
  In order to prove \eqref{eq:I_dchoice_w_clustering_or_cyclic_convergence_as_smalld}, let us now suppose that content access capacity at each node is $C$, in which case the sufficient and necessary conditions for stability are respectively written as $M^{(c)}_{n, d} \cdot \Sigma \leq d C$ and $M^{(c)}_{n, d} \cdot \Sigma \leq 2d C$.
  Using these we find that $C \geq M^{(c)}_{n, d} \cdot \Sigma/d$ is sufficient and $C \geq M^{(c)}_{n, d} \cdot \Sigma/2d$ is necessary for system stability.
  This means that the maximum load on any node in the system will lie in $[M^{(c)}_{n, d} \cdot \Sigma/2d, \; M^{(c)}_{n, d} \cdot \Sigma/d]$, which implies that the load imbalance factor $\mathcal{I}$ for the system lies in $[M^{(c)}_{n, d} \cdot n/2d, \; M^{(c)}_{n, d} \cdot n/d]$.
  
  Finally using the results of almost sure convergence given for $M_{n, d}$ in Lemma~\ref{lm:Md_convergence_as} (hence given for $M^{(c)}_{n, d}$ as well due to Lemma~\eqref{lm:Mcirculard_to_Md_as}), we find \eqref{eq:I_dchoice_w_clustering_or_cyclic_convergence_as_smalld} and \eqref{eq:I_dchoice_w_clustering_or_cyclic_convergence_as_largerd}.
\end{proof}

\subsection{Proof of Lemma~\ref{lm:bibd_neccsuffcond_for_stability}}
\label{subsec:proof_lm_bibd_neccsuffcond_for_stability}
\begin{proof}
  We use the following fact, which we refer to as \textbf{F} here: in a storage allocation with block design, every pair of objects overlaps at \emph{exactly} one node in their choices.
  
  \noindent
  \textbf{Necessary condition}:
  Expansion of a set $S$ of $d$ objects is maximized (of size $d^2$) when the choices for each object are pairwise disjoint. This is not possible due to \textbf{F}.
  Let us start forming $S$ by picking an arbitrary object $o_i$ with the set of choices $C_i$. In order to maximize the expansion of $S$, let us form the rest of $S$ by selecting one object from each node in $C_i$. Given \textbf{F}, no pair in $S \setminus o_i$ is hosted on the same node.
  However, this does not prevent all objects within $S \setminus o_i$ to be hosted on some other node (since a node hosts $d$ different objects).
  In this case the expansion of $S$ will consist of $d + (d-1) + (d-2)^2 = d^2 - 2d + 3$ nodes, which gives us the necessary condition for stability.
  
  \noindent
  \textbf{Sufficient condition}:
  We here consider the spiky load scenario discussed in the proof of Lemma~\ref{lm:rgap_neccsuffcond_for_stability}; let $x$ be an integer in $[0, n]$, and the offered load for $o_i$ is $\rho$ if $i = x + d j$ for some $j = 0, 1, \dots, \floor{n/d}$ and $0$ otherwise.
  Let us refer to objects with spiky load as ``a spiky object''.
  Each spiky object shares its $d$ choices with every other spiky object, and the worst case sharing is when the object has to share $d-1$ of its choices with others. In the worst case, the system is stable only if $\rho \leq 1 + (d-1)/2$, which gives us the sufficient condition for stability.
\end{proof}

\subsection{Proof of Theorem~\ref{thm:P_I_dchoice_wxors}}
\label{subsec:proof_thm_P_I_dchoice_wxors}
\begin{proof}
  This proof is very similar to that of Theorem~\ref{thm:P_I_dchoice_w_clustering_cyclic}, except for the complication that there is no cyclic equivalence of regular balanced $d$-choice allocation with XOR's unlike the case in allocations with object replicas.
  That is why we first find auxiliary cyclic allocations that serve as lower or upper bound on the load balancing ability of the $d$-choice allocation with $r$-XOR's (this is what makes the proof more difficult), then we derive our results by studying these auxiliary cyclic allocations.

  We start by showing sufficient and necessary conditions for system stability.
  
  \vspace{1ex}
  \noindent
  \textbf{\textbf{(i)} Sufficient condition for system stability: }
  This part consists of three intermediate steps.
  
  \vspace{1ex}
  \noindent
  \textit{\underline{Step 1}: Cyclic allocation with $r$-XOR's.}
  Consider a \textit{cyclic} $d$-choice allocation in which for each object $o_i$ that is primarily stored on node $s_i$, $d-1$ choices (recovery sets) are formed by the $d-1$ consecutive disjoint $r$-sets of nodes that come right after $s_i$ (in the order of node indices, by wrapping around the sequence of nodes if necessary).
  For instance, in 3-choice cyclic allocation over nodes $[1, \dots, 6]$ with $r=2$, pairs of nodes $(s_5, s_6)$ and $(s_1, s_2)$ can jointly serve the object $o_4$ that is primarily stored on $s_4$ (recall that we assume the total number of stored objects $k$ is equal to the total number of storage nodes $n$).
  
  Notice that a cyclic allocation cannot be implemented with XOR's. This is because an additional $r$-XOR'ed copy adds a new choice simultaneously for $r$ objects over a set of $r+1$ nodes, and it is not possible for all these added choices to be a proper cyclic choice.
  For instance, let objects $a$, $b$ and $c$ be stored on nodes $s_1$, $s_2$ and $s_3$ respectively, and let us store $a+b$ on $s_3$. Then, $s_2$ and $s_3$ form a choice for $a$, which is a proper cyclic choice, while $s_1$ and $s_3$ form a choice also for $b$ and this is improper for a cyclic allocation, which we simply refer to as a non-cyclic choice.
  However, it is still possible to create an allocation that implements both cyclic and non-cyclic choices with XOR's, then restrict it to behave as a cyclic allocation as follows.
  Firstly, each of the $d-1$ cyclic recovery choices can be created for each object via a separate XOR'ed copy. These XOR'ed copies will incur non-cyclic choices as discussed, but we will ignore and never use them for object access.
  For instance in the previous example, the incurred non-cyclic choice implemented by $a$, $a+b$ stored on $(s_1, s_3)$ can be ignored and never used to access $b$, while a new proper cyclic choice can be added for $b$ by storing $b+c$ on $s_4$.
  In the following we use cyclic allocation, which is created with the restriction described here, merely as a tool to derive our results.
  
  \vspace{1ex}
  \noindent
  \textit{\underline{Step 2}: Cyclic achieves smaller capacity region than non-cyclic.}
  The capacity region of non-cyclic (our regular balanced) allocation with XOR's contains that of its cyclic counterpart.
  To see this is true, we give the following explanation.
  In non-cyclic $d$-choice allocation with $r$-XOR's, each node participates in \textit{at most} $k \cdot d/r$ different choices, while in its cyclic counterpart, each node participates in \textit{exactly} $k \cdot d/r$ different choices.
  In other words, non-cyclic allocation is using the capacity at the nodes more efficiently than its cyclic counter part, while implementing the same number of choices for each object. This expands the capacity region everywhere, or keeps it the same at worst.
  To better understand this, consider the following example of a 2-choice allocation with 2-XOR's and its corresponding allocation matrix
  \[ \icol{a\\e+f}, ~\icol{b}, ~\icol{c\\a+b}, ~\icol{d}, ~\icol{e\\c+d}, ~\icol{f} \]
  \setcounter{MaxMatrixCols}{20}
  \begin{equation}
    \bm{M} = \begin{bmatrix}
      1 & 0 & 0 & 0 & 0 & 0 & 0 & 0 & 0 & 1 & 0 & 1 \\
      0 & 1 & 1 & 0 & 0 & 0 & 0 & 0 & 0 & 0 & 0 & 0 \\
      0 & 1 & 0 & 1 & 1 & 0 & 0 & 0 & 0 & 0 & 0 & 0 \\
      0 & 0 & 0 & 0 & 0 & 1 & 1 & 0 & 0 & 0 & 0 & 0 \\
      0 & 0 & 0 & 0 & 0 & 0 & 0 & 1 & 1 & 0 & 0 & 0 \\
      0 & 0 & 0 & 0 & 0 & 0 & 0 & 0 & 0 & 1 & 1 & 0
    \end{bmatrix}.
  \label{eq:dchoice_wxors_eg}
  \end{equation}

  We next briefly explain what $\bm{M}$ represents. 
  The system achieves stability by splitting (balancing) the demand for each object across its $d$ choices in such a way that no node is over burdened (i.e., each node is offered a load of $<1$).
  Each service choice for an object is either implemented by a single (primary) node or jointly by $r$ nodes (an XOR'ed choice).
  The portion of an object's demand that is forwarded to and supplied by one of its XOR'ed choices flows simultaneously into the $r$ nodes that jointly implement the choice.
  Each $1$ within the $i$th row of $\bm{M}$ represents the assignment of an object's demand portion to node $s_i$.
  For instance, $s_1$ implements the first (primary) choice for $a$ (hence the first $1$ in the $1$st row), and participates in the second choice for objects $e$ and $f$ (hence the second and third $1$ in the $1$st row).
  The cyclic counterpart of the allocation given above in \eqref{eq:dchoice_wxors_eg} would be
  \[ \icol{a\\e+f}, ~\icol{b\\f+a}, ~\icol{c\\a+b}, ~\icol{d\\b+c}, ~\icol{e\\c+d}, ~\icol{f\\d+e} \]
  \begin{equation}
    \bm{M}^c = 
    \begin{bmatrix}
      1 & 0 & 0 & 0 & 0 & 0 & 0 & 0 & 0 & 1 & 0 & 1 \\
      0 & 1 & 1 & 0 & 0 & 0 & 0 & 0 & 0 & 0 & 0 & \underline{1} \\
      0 & 1 & 0 & 1 & 1 & 0 & 0 & 0 & 0 & 0 & 0 & 0 \\
      0 & 0 & 0 & \underline{1} & 0 & 1 & 1 & 0 & 0 & 0 & 0 & 0 \\
      0 & 0 & 0 & 0 & 0 & \underline{1} & 0 & 1 & 1 & 0 & 0 & 0 \\
      0 & 0 & 0 & 0 & 0 & 0 & 0 & \underline{1} & 0 & 1 & 1 & 0
    \end{bmatrix},
  \label{eq:dchoice_wxors_eg_cyclic}
  \end{equation}
  where the incurred non-cyclic choices are ignored in $\bm{M}^c$, e.g., $a$ is never accessed from $f+a$ and $f$.
  Each row sums to $3$ in $\bm{M}^c$, while half the rows sum to less than $3$ in $\bm{M}$ (notice the additional $\underline{1}$'s in $\bm{M}^c$).
  In other words, non-cyclic allocation implements the same number of choices for each object with smaller overlap between different choices compared to its cyclic counterpart.
  A node's capacity is shared by all the service choices in which the node participates. Thus, it is better to have less overlap between choices in order to achieve greater capacity region.
  This is the ``inefficiency'' of cyclic allocation that causes it to achieve a smaller capacity region than its non-cyclic counterpart.
  
  Let $D = \left\{\bm{x}~|~\bm{M} \cdot \bm{x} \preceq \bm{1}, ~\bm{x} \succeq \bm{0}\right\}$, and $D^c$ be defined similarly with $\bm{M}^c$.
  It is easy to see that any $\bm{x}$ in $D^c$ will also lie in $D$ (recall the additional $\underline{1}$'s in $\bm{M}^c$). In addition, 
  non-cyclic $d$-choice allocation and its cyclic counterpart share the same $\bm{T}$ (i.e., the other allocation matrix that yields the capacity region $\mathcal{C}$ (or $\mathcal{C}^c$) by transforming $D$ (or $D^c$); see Sec.~\ref{subsec:cap_region}).
  Thus, we have $\mathcal{C} \supseteq \mathcal{C}^c$.
  This together with Def.~\ref{def:P_Sigma} implies that probability $\mathcal{P}_{\Sigma}$ for non-cyclic $d$-choice allocation is at least as large as that for its cyclic counterpart.
  
  \vspace{1ex}
  \noindent
  \textit{\underline{Step 3}: A sufficient condition for the stability of cyclic allocation.}
  Recall from Lemma~\ref{lm:rgap_neccsuffcond_for_stability} how we found a sufficient condition for stability when the allocation is constructed with the clustering or cyclic ($r$-gap in general) designs.
  Using the same arguments, a sufficient condition for the stability of the cyclic $d$-choice allocation with $r$-XOR's is found as $M^{(c)}_{n, 1 + r(d-1)} \cdot \Sigma \leq d$, where $\Sigma$ is the cumulative offered load on the system and $M^{(c)}_{n, 1 + r(d-1)}$ is the maximal ($1 + r(d-1)$)-spacing for $n$ uniform spacings on the unit circle.
  The reason for caring about ($1 + r(d-1)$)-spacing's in this case (rather than $d$-spacing's as was the case for allocations with object replicas) is because an object's first choice is implemented by the (primary) node that stores the object, and its XOR'ed choices are implemented by the $d-1$ disjoint $r$-sets of nodes that follow the primary node in (cyclic) order.
  The reason for keeping the right hand side of the sufficient condition unchanged at $d$ (as for the allocations with replication) is that object access from an $r$-XOR'ed choice requires accessing all $r$ nodes that jointly implement the choice, so $r(d-1)$ nodes that form the $d-1$ XOR'ed choices for an object can at most provide a capacity of $d-1$, which together with the capacity of the primary node adds up to $d$.
  
  \vspace{1ex}
  \noindent
  \textit{\underline{Final Step}: Putting it all together.}
  As discussed above, $\mathcal{P}_{\Sigma}$ for cyclic $d$-choice allocations is a lower bound for that of its non-cyclic counterpart.
  Thus, the sufficient condition $M^{(c)}_{n, 1 + r(d-1)} \cdot \Sigma \leq d$ for the stability of cyclic allocations will also be sufficient for the stability of its non-cyclic counterpart (i.e., our regular balanced allocation).
  
  \vspace{1ex}
  \noindent
  \textbf{\textbf{(ii)} Necessary condition for system stability:}
  We again here relate the cyclic allocation with $r$-XOR's as introduced in part (i) to its non-cyclic counterpart (our regular balanced allocation).
  We do this again in three intermediate steps that are in the same spirit as those given in part (i).
  
  \vspace{1ex}
  \noindent
  \textit{\underline{Step 1}: Cyclic-plus allocation.}
  Recall in part (i) that we created a cyclic $d$-choice allocation by adding all $k(d-1)$ XOR'ed copies that are necessary to implement the $d-1$ cyclic choices for each object; and then ignoring the incurred non-cyclic choices by never considering them for object access (e.g., recall the non-cyclic allocation in \eqref{eq:dchoice_wxors_eg} and its cyclic counterpart in \eqref{eq:dchoice_wxors_eg_cyclic}).
  Let us also consider and use the incurred non-cyclic choices for object access here, and refer to this form of allocation as \textit{cyclic-plus}.
  
  \vspace{1ex}
  \noindent
  \textbf{Cyclic-plus achieves greater capacity region than non-cyclic.}
  The capacity region of cyclic-plus allocations will contain that of its non-cyclic counterpart, which together with Def.~\ref{def:P_Sigma} implies that $\mathcal{P}_{\Sigma}$ for cyclic-plus allocations will be at least as large as that of its non-cyclic counterpart.
  This is because cyclic-plus allocations implement all the choices that their non-cyclic counterparts implement plus some additional choices (e.g., compare \eqref{eq:dchoice_wxors_eg} with \eqref{eq:dchoice_wxors_eg_cyclic}), which will yield at least as large a capacity region everywhere as the one without the additional choices.
  
  \vspace{1ex}
  \noindent
  \textit{\underline{Step 2}: A necessary condition for the stability of cyclic-plus allocation.}
  In cyclic-plus $d$-choice allocations, there are $d$ cyclic and $d$ non-cyclic choices for each object. Notice that each non-cyclic choice for an object is due to a cyclic choice of another object.
  Consider an object primarily stored on $s_i$, then $s_{i+1 \mod n}$ participates in this object's first cyclic choice and all of its non-cyclic choices. This is a direct consequence of how cyclic choice with XOR's are constructed.
  For instance, consider object $a$ in \eqref{eq:dchoice_wxors_eg_cyclic}, its first cyclic choice is ($b$, $a+b$) and its only non-cyclic choice is ($f+a$, $a$), where both $b$ and $f+a$ are stored on the node that comes right after $a$'s primary node.
  Thus, all of the $d$ additional non-cyclic choices and the first cyclic choice for an object depend on a single node, which will be a bottleneck when these choices must be used simultaneously to access the object. In other words, all of these $d+1$ choices (one cyclic and $d$ non-cyclic) can simultaneously yield at most as much capacity as of a single node.
  
  Due to the bottleneck node described above, even the additional non-cyclic choices are not sufficient to achieve stability in a cyclic-plus $d$-choice allocation with $r$-XOR's when any $1 + r(d-1)$ consecutive nodes have a cumulative offered load $> 2d-1$, that is when $M^{(c)}_{n, 1 + r(d-1)} \cdot \Sigma > 2d-1$, hence a necessary condition for stability is that as $M^{(c)}_{n, 1 + r(d-1)} \cdot \Sigma \leq 2d$.
  This is easy to see using the exact same arguments we used to show the corresponding necessary stability condition in Lemma~\ref{lm:rgap_neccsuffcond_for_stability} for $d$-choice allocation with object replicas.
  
  \vspace{1ex}
  \noindent
  \textit{\underline{Final Step}: Putting it all together.}
  We showed that the probability $P_{\Sigma}$ for cyclic-plus allocation is an upper bound on that of its non-cyclic counterpart, thus $M^{(c)}_{n, 1 + r(d-1)} \cdot \Sigma \leq 2d$ is also a necessary stability condition for non-cyclic (our regular balanced) $d$-choice allocation.
  
  From now on we will refer to $1 + r(d-1)$ as $D$.

  \vspace{1ex}
  \noindent
  \textbf{Proof of \eqref{eq:P_dchoice_wxors_convergence_smalld} and \eqref{eq:P_dchoice_wxors_convergence_largerd}.}
  Follows from the exact same arguments used in the proof of respectively \eqref{eq:P_dchoice_w_clustering_cyclic_convergence_smalld} and \eqref{eq:P_dchoice_w_clustering_cyclic_convergence_largerd} (Theorem~\ref{thm:P_I_dchoice_w_clustering_cyclic}).
  
  \vspace{1ex}
  \noindent
  \textbf{Proof of \eqref{eq:I_dchoice_wxors_convergence_as_smalld} and \eqref{eq:I_dchoice_wxors_convergence_as_largerd}.}
  Using the same arguments used in the proof of \eqref{eq:I_dchoice_w_clustering_or_cyclic_convergence_as_smalld} and \eqref{eq:I_dchoice_w_clustering_or_cyclic_convergence_as_largerd} (Theorem~\ref{thm:P_I_dchoice_w_clustering_cyclic}),
  we can conclude here that the load imbalance factor $\mathcal{I}$ for the system lies in $[M^{(c)}_{n, D} \cdot n/2d, \; M^{(c)}_{n, D} \cdot n/d]$.
  The results of almost sure convergence given for $M_{n, D}$ in Lemma~\ref{lm:Md_convergence_as} hold also for $M^{(c)}_{n, D}$ by Lemma~\ref{lm:Mcirculard_to_Md_as}. Using the same arguments given in the proof of Lemma~\ref{lm:Md_convergence_as}, we derive \eqref{eq:I_dchoice_wxors_convergence_as_smalld} and \eqref{eq:I_dchoice_wxors_convergence_as_largerd}.
\end{proof}

\section{Acknowledgements}
Authors would like to thank Gala Yadgar for the fruitful discussions and recommendations in the early stages of this work.
This research is supported by the National Science Foundation
under Grants No. CIF-1717314.

\bibliographystyle{IEEEtran}
\bibliography{main}


\begin{IEEEbiographynophoto}{Mehmet Fatih Aktas}
is currently a senior engineer at MathWorks. His main research interest is to make distributed computer systems faster and more robust to runtime changes. He has been doing research in both systems development and theoretical analysis. In his research, he mostly relies on probabilistic modeling and on tools from applied probability such as queueing theory, order statistics and reinforcement learning. He previously completed an MS and a PhD in Electrical and Computer Engineering at Rutgers University, New Brunswick, NJ, USA, and a BS in Electrical and Electronics Engineering at Bilkent University, Ankara, Turkey.
\end{IEEEbiographynophoto}

\begin{IEEEbiographynophoto}{Amir Behrouzi-Far}
received the B.S. degree in electrical engineering from Iran University of Science and Technology, Tehran, Iran, in 2013 and M.S. degree in electrical and electronics engineering from Bilkent University, Ankara Turkey, in 2016. He is currently pursuing the Ph.D. degree in electrical and computer engineering at Rutgers University, New Brunswick, NJ, USA.
His research interest includes performance evaluation in distributed systems, timeliness in real time systems, scheduling in computing/storage systems, wireless communications and reinforcement learning.
\end{IEEEbiographynophoto}


\begin{IEEEbiographynophoto}{Emina Soljanin}
is a professor at Rutgers University. Before moving to Rutgers in January 2016, she was a (Distinguished) Member of Technical Staff for 21 years in the Mathematical Sciences Research of Bell Labs.%
Her interests and expertise are wide. Over the past quarter of the century, she has participated in numerous research and business projects, as diverse as power system optimization, magnetic recording, color space quantization, hybrid ARQ, network coding, data and network security, distributed systems performance analysis, and quantum information theory. She served as an Associate Editor for Coding Techniques, for the IEEE Transactions on Information Theory, on the Information Theory Society Board of Governors, and in various roles on other journal editorial boards and conference program committees.
 Prof.~Soljanin an IEEE Fellow, an outstanding alumnus of the Texas A\&M School of Engineering, the 2011 Padovani Lecturer, a 2016/17 Distinguished Lecturer, and 2019 President of the IEEE Information Theory Society.
\end{IEEEbiographynophoto}

\begin{IEEEbiographynophoto}{Phil Whiting}
received his BA degree from the University of Oxford, his MSc from the University of London and his Ph. D. was in queueing
theory from the University of Strathclyde. After a post-doc at the University of Cambridge, Phil's interests centered on wireless. In 1993 Phil participated in the Telstra trial of Qualcomm CDMA in South Eastern Australia. He then joined the Mobile research Centre at the University of South Australia Adelaide. He was a researcher at Bell Labs from January 1997 to June 2013. Phil is currently a research Professor at Macquarie University, Sydney Australia and also a consultant to Telstra for the past two years. Phil has over 25 patents in applications for  DSL vectoring, Wireless Networks and Location and Tracking. Phil has several awards including for his work in DSL vectoring and in wireless scheduling. 
 
Phil has held visiting positions including ones at Brown University (Maths), University of Korea (Engineering) and Vrij University (Maths). For the past 3 years Phil has been a STAR visiting scholar to the Maths Dept., Technical University of Eindhoven, this collaborative work includes investigations of both CSMA networks and load balancing in queueing networks.

Apart from papers in various aspects of telecommunications, Phil has been author on various aspects of probability theory, including random Vandermonde matrices, Large Deviations theory for Occupancy Models and more recently random CSMA networks and load balancing in queueing networks. Phil's current research includes Storage Systems and Wireless mmWave networks.
\end{IEEEbiographynophoto}

\end{document}